\documentclass[american,aps,prx,reprint,superscriptaddress, longbibliography,floatfix,appendix]{revtex4-2}

\usepackage{inputenc}
\usepackage{bbold}
\usepackage{color}
\usepackage{babel}
\usepackage{mathrsfs}
\usepackage{physics}
\usepackage{array} 
\usepackage{mathtools}
\usepackage{bbding}
\usepackage{pifont}
\usepackage{textcomp}
\usepackage{wasysym}
\usepackage{amsthm}
\usepackage{nicefrac}
\usepackage{wasysym}
\usepackage{amsmath}
\usepackage{bbm}
\usepackage{amssymb}
\usepackage{graphicx}
\usepackage{enumitem}
\usepackage{xcolor}

\usepackage[colorlinks=true,linkcolor=blue,citecolor=red,plainpages=false,pdfpagelabels]{hyperref}
\usepackage{booktabs}   
\usepackage{tabularx}

\newtheorem{theorem}{Theorem}
\newtheorem*{theorem*}{Theorem}

\newtheorem{proposition}{Proposition}
\newtheorem*{proposition*}{Proposition}
\newtheorem{corollary}{Corollary}
\newtheorem{lemma}{Lemma}
\newtheorem*{example*}{Example}

\newtheorem*{remark*}{Remark}
\newtheorem{definition}{Definition}
\newtheorem*{lemma*}{Lemma}

\def\tr{\operatorname{tr}}

\def\St{\operatorname{St}}
\def\Ch{\operatorname{Ch}}
\def\SCh{\operatorname{SCh}}
\def\Pos{\operatorname{Pos}}
\def\id{\operatorname{id}}

\def\n{\mathcal{N}}
\newcommand{\mc}[1]{\mathcal{#1}}

\begin{document}

\title{Thermodynamics of quantum processes: An operational framework for free energy and reversible athermality}

\author{Himanshu Badhani}\email{himanshubadhani@gmail.com}
\affiliation{Center for Security, Theory and Algorithmic Research, International Institute of Information Technology, Hyderabad, Gachibowli, Telangana 500032, India}
\affiliation{Centre for Quantum Science and Technology, International Institute of Information Technology, Hyderabad, Gachibowli, Telangana 500032, India}

\author{Dhanuja G.S.}
\email{dhanuja.g@research.iiit.ac.in}
\affiliation{Centre for Quantum Science and Technology, International Institute of Information Technology, Hyderabad, Gachibowli, Telangana 500032, India}
\affiliation{Center for Computational Natural Sciences and Bioinformatics, International Institute of Information Technology, Hyderabad, Gachibowli, Telangana 500032, India}

\author{Siddhartha Das}
\email{das.seed@iiit.ac.in}
\affiliation{Center for Security, Theory and Algorithmic Research, International Institute of Information Technology, Hyderabad, Gachibowli, Telangana 500032, India}
\affiliation{Centre for Quantum Science and Technology, International Institute of Information Technology, Hyderabad, Gachibowli, Telangana 500032, India}

\begin{abstract}
We explore the thermodynamics of quantum processes (quantum channels) by axiomatically introducing the free energy for channels, defined via the quantum relative entropy with an absolutely thermal channel whose fixed output is in equilibrium with a thermal reservoir. This definition finds strong support through its operational interpretations in designated quantum information and thermodynamic tasks. We construct a resource theory of athermality for quantum processes, where free operations are Gibbs preserving superchannels and golden units are unitary channels with respect to absolutely thermal channel having fully degenerate output Hamiltonian. We exactly characterize the one-shot distillation and formation of quantum channels using hypothesis-testing and max-relative entropy with respect to the absolutely thermal channel. These rates converge asymptotically to the channel free energy (up to a multiplicative factor of half the inverse temperature), establishing its operational meaning and proving the asymptotic reversibility of the athermality. We show the direct relation between the resource theory of athermality and quantum information tasks such as private randomness and purity distillation, and thermodynamic tasks of erasure and work extraction. Our work connects the core thermodynamic concepts of free energy, energy, entropy, and maximal extractable work of quantum processes to their information processing capabilities.
\end{abstract}

\maketitle
\section{Introduction}

Quantum computation and information processing are fundamentally built upon quantum processes~\cite{LJL+10,Das19,DBWH21,PRY+22,SSHD24}, which are the physical transformations and evolutions that quantum systems undergo. Studying the thermodynamic aspects of these systems and processes through an information-theoretic approach offers essential insights into the fundamental limitations on preparation, manipulation, and measurement of quantum resources, especially when faced with energetic constraints~\cite{Ho04,Pek15,CC22,SSC23,SDC21,HWS+25}. Free energy stands as a core concept in thermodynamics, and its extension to quantum states provides insight into limitations on allowed transformations of states~\cite{BHN+15,LKJ+15,SS21,MCZG24,ST25} and possible utilization for thermodynamic engines and circuits~\cite{SU08,CTH09,NG15,ULK15}. 

\begin{table*}[htbp]
    \centering
    \renewcommand{\arraystretch}{1.0} 
    \setlength{\extrarowheight}{0pt}

    \begin{tabularx}{\textwidth}{
        >{\color{violet}\raggedright\arraybackslash}p{0.25\textwidth}
        >{\color{magenta}\raggedright\arraybackslash}p{0.34\textwidth}
        >{\color{blue}\raggedright\arraybackslash}p{0.36\textwidth}
    }
        \toprule
        \textbf{\vspace{-3pt}Concepts\vspace{-0.05pt}} & \textbf{\vspace{-3pt}States\vspace{-0.05pt}} & \textbf{\vspace{-3pt}Channels\vspace{-0.05pt}} \\
        \midrule
        Equilibration with bath & Thermal state $\gamma^\beta$  & Absolutely thermal channel $\mc{T}^\beta$\\
        Energy & $E(\rho)=\langle\widehat{H}\rangle_{\rho}$ & $E[\n]=\sup_{\rho\in\St(A')}\langle \widehat{H}\rangle_{\mc{N}(\rho)}$ \newline (for channel output noninteracting with its reference)\\
        Entropy & $S(\rho)=-\tr[\rho\ln\rho]=-D(\rho\|\mathbbm{1})$ & $S[\n]=-D[\n\Vert\mc{R}^{\mathbbm{1}}]$ \\
        Thermal free energy & $F_{\rm T}^\beta(\rho)=\beta^{-1}D(\rho\|\widehat{\gamma}^\beta)$ & $F_{\rm T}^\beta[\n]=\beta^{-1}D[\n\|\widehat{\mc{T}}^\beta]$\\
        Resource-theoretic free energy & $F^\beta(\rho)=F^\beta_{\rm T}(\rho)-F^\beta_{\rm T}(\gamma^\beta)=\beta^{-1}D(\rho\Vert{\gamma}^\beta)$ & $F^\beta[\n]=F_{\rm T}^\beta[\n]-F_{\rm T}^\beta[\mc{T}^\beta]=\beta^{-1}D[\n\|{\mc{T}}^\beta]$ \\
       
     Minimal free energy & Thermal state $\gamma^\beta$ & Absolutely thermal channel $\mc{T}^\beta$\\
        Maximal extractable work & $W^{\rm ext}_{\rho\to\gamma}=F^\beta_{\rm T}(\rho)-F_{\rm T}^\beta(\gamma^\beta)=F^\beta(\rho)$ & $W^\beta_{\rm ext}[\mc{N}]=F_{\rm T}^\beta[\n]-F_{\rm T}^\beta[\mc{T}^\beta]=F^\beta[\n]$\\
       Helmholtz equation & $F^\beta_{\rm T}(\rho)=E(\rho)-\beta^{-1}S(\rho)$ & $F^{\beta}_{\rm T}[\n]\leq E[\n]-\beta^{-1}S[\n]$ \newline (for channel output noninteracting with its reference; inequality saturates for replacer channels)\\
        \bottomrule
    \end{tabularx}
    \caption{We summarize and compare mathematical expressions for elementary thermodynamic concepts for quantum states vs channels. We consider the thermal reservoir to be at inverse temperature $\beta$ and $\widehat{H}$ denotes the Hamiltonian of the designated system. $\widehat{\gamma}^\beta$ is a thermal (Gibbs) operator and $\gamma^\beta$ is a thermal state. $\mc{R}^\omega_{A'\to A}(\cdot):=\tr[\cdot]\omega_A$, $\mc{T}^\beta_{A'\to A}(\cdot):=\tr[\cdot]\gamma^\beta_A$, $\widehat{\mc{T}}^\beta_{A'\to A}(\cdot):=\tr[\cdot]\widehat{\gamma}^\beta$. $D(\rho\|\sigma):=\lim_{\varepsilon\to 0^+}\tr[\rho(\ln\rho-\ln(\sigma+\varepsilon\mathbbm{1}))]$ is (Umegaki) relative entropy between states $\rho,\sigma$ and $D[\mc{N}\|\mc{M}]:=\sup_{\psi\in\St(RA')}D(\id_R\otimes\mc{N}(\psi_{RA'})\|\id_R\otimes\mc{M}(\psi_{RA'}))$ is relative entropy between channels $\mc{N}_{A'\to A},\mc{M}_{A'\to A}$. We assume $\rho$ is a density operator of system $A'$ and $\mc{N}$ is a quantum channel from $A'\to A$. The content for quantum states are known in literature~\cite{BHN+15,BHO+13,DC19,Gou24}, prior to our current work. The expression of the entropy of quantum channels is from \cite{GW21,SPSD25}. See Section~\ref{sec:outline} for the outline of the paper.}
    \label{tab:concepts}
\end{table*}

Motivated by fundamental interests and technological imperative of miniaturization, we explore the thermodynamics of quantum processes and their implications for information processing capabilities. Quantum processes are formally described by quantum channels, which are completely positive and trace-preserving (CPTP) maps. Quantum channels capture all physically realizable transformations of quantum states, including quantum measurements and even states themselves, when regarded as preparation channels. Although the action of a channel can be inferred from how it transforms input states, the information-theoretic properties established for quantum states do not always extend directly to channels. Operational tasks such as channel discrimination and resource conversion reveal richer and more intricate structures at the level of quantum channels~\cite{DW19,KDWW19,FFRS20,LXD+25,DGP24}. However, many information-theoretic results for channels naturally recover the corresponding results for states as special cases, since quantum states can be regarded as preparation channels.

Our objective is the quantification and manipulation of the thermodynamic resources intrinsic to quantum channels \cite{SPSD25,DS25}. The thermodynamic characterization of general quantum channels necessitates the development of new conceptual frameworks. In this work, we use an axiomatic approach~\cite{BHO+13,BHN+15,BRL+19} to define a notion of free energy for quantum channels, which enables us to formulate a resource theory of thermodynamic processes where this notion of free energy acquires a clear operational meaning. This framework also enables the analysis of how core thermodynamic quantities such as free energy, entropy, and energy appear at the level of quantum processes, and how they constrain or enable information processing tasks. This unified view establishes formal links between the thermodynamics of quantum channels and operational tasks such as randomness extraction, purity distillation, and work extraction~\cite{NG15,SSP15,GMN+15,YHW19,YZGZ20,MEC+19,BGC+25}.

The process of thermalization, which describes non-unitary evolution toward equilibrium with an environment (thermal bath) \cite{HP93,Ger93,Tas00,BHO+13,MHK23,BC25}, serves as a concrete example of a thermodynamic quantum process. It is characterized by the irreversible loss of coherence and the emergence of thermal states that minimize the system's free energy~\cite{GLTZ06,LPSW09,HMG19,DS25}. In this context, the concepts of free energy and the resource theory of athermality for quantum channels provide a rigorous framework for quantifying thermodynamic costs and understanding the fundamental limitations of quantum operations, thereby offering a clear connection between quantum information theory and foundational thermodynamic principles.

\textit{Note added}. This paper is companion to ``Thermodynamic work capacity of quantum information processing" \cite{BGD25b}.

\subsection{Main Results}
The concept of free energy is fundamental and instrumental in thermodynamics. We extend this notion to quantum processes which describe physical transformations of quantum systems in quantum theory. We leverage  axiomatic and information-theoretic approaches to define meaningful free energy functions of quantum channels (Theorem~\ref{thm:axioms-satisifaction}) that find operational interpretations in quantum information processing and thermodynamical tasks. We establish relations between the elementary thermodynamic quantities like free energy, energy, entropy, and operational quantities like extractable work to shed light on the thermodynamics of quantum processes (see Table~\ref{tab:concepts}). 

We define the resource-theoretic generalized free energy functions $\mathbf{F}^\beta[\n]$ of a quantum channel $\n_{A'\to A}$ as the inverse temperature times the appropriate generalized channel divergence functions, $\mathbf{F}^\beta[\n]:=\beta^{-1}\mathbf{D}[\n\Vert\mc{T}^\beta]$~\cite{SPSD25,DGP24}, between the channel $\n$ and an absolutely thermal channel $\mc{T}^\beta_{A'\to A}$. The channel $\mc{T}^\beta_{A'\to A}$  thermalizes all its inputs such that the output state is a thermal state $\gamma^\beta_A$ in equilibrium with the bath at an inverse temperature $\beta>0$, $\mc{T}^\beta(\cdot)=\tr[\cdot]\gamma^\beta$. Some appropriate generalized channel divergences for our purpose are shown to be the Umegaki relative entropy and max-relative entropy. We also discuss smoothed versions of free energy functions derived from these channel divergences and the $\varepsilon$-hypothesis-testing relative entropy.

We develop the dynamical resource theory of athermality for (square) quantum channels. The absolutely thermal channel $\mc{T}^\beta$ is considered a free object, where $\mc{T}^\beta$ depends on the inverse temperature of the bath and the Hamiltonian of the output system under consideration. This is a natural assumption because of the inherent thermalization process that equilibrates quantum systems with the bath. Even near isolation of quantum systems from the bath for a short time requires engineering and interventions, and perfect closed quantum dynamics are idealistic. The free operations are Gibbs preserving superchannels $\Theta^\beta$ which map absolutely thermal channel $\mc{T}^\beta_{A'\to A}$ to absolutely thermal channel $\mc{T}^{\beta}_{B'\to B}$ of the same inverse temperature $\beta$. A quantum channel $\n$ is athermal if it is not absolutely thermal $\mc{T}^\beta$ and is deemed resourceful. For clarity,  we use $(\n,\mc{T}^\beta)$ to discuss the thermodynamic resourcefulness of $\n$ with respect to $\mc{T}^\beta$~\cite{DS25}.

The resource-theoretic (and thermodynamic) free energy and max-free energy are maximum for a square quantum channel if and only if it's a unitary channel (Proposition~\ref{thm:unitary-free-energy}). Among all possible Hamiltonians, the free energy of a unitary channel $\mc{U}_{A'\to A}$ is least if the output Hamiltonian is fully degenerate, $\widehat{H}_A=c\mathbbm{1}_A$ for $c\in\mathbbm{R}$. As the Hamiltonian can be arbitrary (and large), we advocate that a golden unit of the dynamical resource theory of athermality is the resource channel $(\id_m,\mc{R}^\pi)$, where ${\n}_m$ denotes that input and output Hilbert spaces of channel $\n$ are $m$-dimensional and $\mc{R}^\pi$ is the uniformly mixing channel that only outputs the maximally mixed state $\pi$. All unitary resource channels $(\mc{U}_m,\mc{R}^\pi)$ are equivalent (equally resourceful) to $(\id_m,\mc{R}^\pi)$. In fact, the interconversion between $(\mc{U},\mc{T}^\beta)\leftrightarrow (\mc{V},\mc{T}^\beta)$ is possible via Gibbs preserving superchannels (free operations) for any pair of unitary channels $\mc{U}$ and $\mc{V}$.

We find that the single-shot athermality distillation $\mathrm{Dist}^\varepsilon (\n,\mathcal{T}^\beta)$ and formation $\mathrm{Cost}^\varepsilon (\n,\mathcal{T}^\beta)$ of a quantum resource channel $(\mc{N},\mc{T}^\beta)$, up to allowed error $\varepsilon\in(0,1)$ under the action of Gibbs preserving superchannels, are $\frac{1}{2}D_{\infty}^\varepsilon[\n\Vert\mathcal{T}^\beta]$ and $\frac{1}{2}D_{H}^{\varepsilon^2}[\n\Vert\mathcal{T}^\beta]$, respectively (Theorem~\ref{thm:dist_cost}). We have $F^{\beta,\varepsilon}_H[\n]=\frac{2}{\beta}\mathrm{Dist}^{\sqrt{\varepsilon}} (\n,\mathcal{T}^\beta)$ and $ F^{\beta,\varepsilon}_{\infty}[\n]=\frac{2}{\beta}\mathrm{Cost}^\varepsilon (\n,\mathcal{T}^\beta)$, which provide exact operational meaning to the resource-theoretic hypothesis-testing free energy and the max-free energy of a quantum channel. The asymptotic athermality distillation and formation rates of a quantum resource channel $(\n,\mc{T}^\beta)$, under parallel uses of channels, are equal to $\frac{1}{2}D[\mc{N}\Vert\mc{T}^\beta]$~(Theorem~\ref{thm:rev}). The resource-theoretic free energy of a channel $\n$,  $F^\beta[\n]=\beta^{-1}D[\mc{N}\Vert\mc{T}^\beta]$, and thus its asymptotic athermality distillation and formation rates are both equal to $\frac{\beta}{2}F^\beta[\n]$, providing the exact operational interpretation to the resource-theoretic free energy of a quantum channel. Both asymptotic athermality distillation and formation rates of a quantum channel are the same, implying that the athermality of quantum channels under Gibbs preserving superchannels is (asymptotically) reversible.

We get the resource theory of purity for quantum channels as a special case in the limit $\beta\to 0^+$ or Hamiltonian of the channel output is trivial. We discuss the dependence of the free energy of a quantum channel on the quantum transmission ability of the channel and free energy of its output (Proposition~\ref{prop:mi}). We elucidate on the novel relations between the free energy of a quantum channel and its energy and entropy. The thermal free energy $F^\beta_{\rm T}[\n]:=\beta^{-1}D[\n\|\widehat{\mc{T}}^\beta]$, where $\widehat{\mc{T}}^\beta$ is a replacer (or preparation) map that always outputs thermal operator $\widehat{\gamma}^\beta$ (an unnormalized thermal state), of a quantum channel $\n$ and its entropy $S[\n]=-D[\n\|\mc{R}^\mathbbm{1}]$~\cite{GW21,SPSD25}, where $\mc{R}^\mathbbm{1}(\cdot)=\tr[\cdot]\mathbbm{1}$, and thermodynamic energy $E[\n]$ satisfy the relation $F^{\beta}_{\rm T}[\n]\leq E[\n]-\beta^{-1}S[\n]$ when the channel output is noninteracting with its reference (Theorem~\ref{thm:fets}).

We prove that the resource-theoretic free energy of a quantum channel $\n$ is the difference of the thermal free energy of the channel and the absolutely thermal channel, $F^\beta[\n]=F^{\beta}_{\rm T}[\n]-F^{\beta}_{\rm T}[\mathcal{T}^\beta]$, and has operational interpretation as the maximum thermodynamic utility of the channel for maximal extractable work (Theorem~\ref{thm:excess}). We believe our framework, results, and discussions will provide deeper insights into the thermodynamical aspects of quantum processes and interconnectedness with their quantum information processing capabilities.

\subsection{Outline}\label{sec:outline}
The paper is organized as follows. We introduce standard notations, definitions, and facts in Section~\ref{sec:prem} that are used to derive main results and discussions surrounding them. In Section~\ref{sec:freeenergy}, we employ an axiomatic approach to define the resource-theoretic free energy functions of a quantum channel and provide reasoning behind axioms. We discuss some notable definitions, namely, the (resource-theoretic) free energy, max-free energy, and hypothesis-testing free energy of quantum channels, and prove some of their important properties. We introduce the resource theory of athermality for square quantum channels in Section~\ref{sec:resourcetheory} and prove that the resource theory of athermality under Gibbs preserving superchannels is asymptotically reversible. We inspect quantum transmission capabilities of quantum channels and its free energy in Section~\ref{sec:aspects}. We discuss the relation between the resource theory of athermality with private randomness distillation and the dynamical resource theory of purity. In Section~\ref{sec:energy}, we delve deeper into the relations between the resource-theoretic free energy, thermal free energy, entropy, and energy of quantum channels. We discuss the operational meaning of the resource-theoretic free energy of a quantum channel with its thermodynamic utility for the maximal extractable work. We finally conclude in Section~\ref{sec:discussion}. The detailed proofs of several results are discussed in the Appendix.

\section{Preliminaries}\label{sec:prem}
We consider quantum systems associated with separable Hilbert spaces with tensor-product structure. For a bipartite state $\rho_{AB}$, $\rho_A=\tr_B(\rho_{AB})$. The energy of a quantum state $\rho_A$ associated with Hamiltonian $\widehat{H}_A$ is $\langle \widehat{H}\rangle_{\rho}:=\tr[\widehat{H}_A\rho_A]$. Let $\St(A)$ and $\Pos(A)$ denote the set of all density operators (quantum states) and the set of all positive semidefinite operators on $A$, respectively. Let $\Ch(A',A)$ denote the set of all quantum channels from $A'$ to $A$. Let $\SCh((A',A),(B',B))$ denote the set of all quantum superchannels mapping $\Ch(A',A)$ to $\Ch(B',B)$. Any quantum superchannel can be written as concatenation of preprocessing and postprocessing quantum channels~\cite{CDP09}, i.e., for any $\Theta\in\SCh((A',A),(B',B))$, there exists preprocessing channel $\mathcal{P}\in\Ch(B',RA')$ and postprocessing channel $\mathcal{Q}\in\Ch(RA,B)$ such that $\Theta(\mathcal{N})=\mathcal{Q}\circ\mathcal{N}\circ\mathcal{P}$. A unitary superchannel is composed of unitary postprocessing and preprocessing channels.

 $\id_R$ denotes the identity channel on $R$ and $\mathbbm{1}_R$ denotes the identity operator on $R$. Let $\Gamma^{\mathcal{N}}_{RA}$ denote a maximally entangled operator, $\Gamma^{\mathcal{N}}_{RA}:=\sum_{i,j=0}^{d-1}\ket{ii}\bra{jj}_{RA}$ for $d=\min\{|R|,|A|\}$; a maximally entangled state $\Phi_{RA}=\frac{1}{d}\Gamma_{RA}$ if $d<\infty$. For a linear supermap $\mc{N}_{A'\to A}$, $\Gamma^{\mc{N}}_{RA}=\id_R\otimes\mc{N}(\Gamma_{RA'})$ denotes its Choi operator. If $\mc{N}\in\Ch(A',A)$ is a finite-dimensional quantum channel, then $\Phi^{\mc{N}}_{RA}=\frac{1}{d}\Gamma^{\mc{N}}_{RA}$ is its Choi state. For $|A|<\infty$, $\pi_A:=\frac{1}{|A|}\mathbbm{1}_{A}$ is the maximally mixed state. $\norm{X}_p$ denotes the Schatten p-norm of an operator $X$. For any two Hermiticity-preserving maps $\mathcal{N},\mathcal{M}$, $\norm{\mc{N}-\mc{M}}_\diamond:=\sup_{\psi\in\St(RA')}\norm{\id_R\otimes\mc{N}(\psi_{RA'})-\id_R\otimes\mc{M}(\psi_{RA'})}_1$. The fidelity between $\rho,\sigma\in\St(A)$ is $F(\rho,\sigma):=\norm{\sqrt{\rho}\sqrt{\sigma}}_1^2$. Let $\langle\widehat{O}\rangle_{\rho_A}:=\tr[\widehat{O}_A\rho_A]$. For a quantum channel $\mc{N}_{A'\to A}$ and a bipartite state $\rho_{RA'}$, $\mc{N}(\rho_{RA'}):=\id_R\otimes\mc{N}(\rho_{RA'})$.

 \textit{Generalized state and channel divergences}.--- Let $\mathbf{D}(\cdot\Vert\cdot)$ be the generalized state divergence, i.e., for an arbitrary pair of quantum states $\rho,\sigma$ and an arbitrary quantum channel $\mathcal{N}$, we have $\mathbf{D}(\rho\Vert\sigma)\geq \mathbf{D}(\mathcal{N}(\rho)\Vert\mathcal{N}(\sigma))$. The generalized channel divergence $\bf{D}[\cdot\Vert\cdot]$ is derived from the generalized state divergence: For any two quantum channels $\mathcal{N},\mc{M}\in\Ch(A'A)$ and the identity channel $\id$,
\begin{equation}
    \mathbf{D}[\mathcal{N}\Vert\mathcal{M}]:=\sup_{\psi\in\mathrm{St}(RA')}\mathbf{D}(\id_R\otimes\mathcal{N}(\psi)\Vert\id_R\otimes\mathcal{M}(\psi)),
\end{equation}
where it suffices to optimize over pure states $\psi_{RA'}$~\cite{LKDW18}. Some examples of the generalized state divergences~\cite{Tom21,WWY14,Dat09,LR12}: For a state $\rho_A$ and a positive semidefinite operator $\sigma_A$, whenever $\operatorname{supp}(\rho) \subseteq \operatorname{supp}(\sigma)$ (otherwise $+\infty$) (a) the quantum relative entropy $D(\rho\Vert\sigma):=\operatorname{tr}\left[\rho (\ln \rho - \ln \sigma)\right]$, (b) the max-relative entropy $D_{\infty}(\rho\Vert\sigma):=\ln \norm{\sigma^{-1/2} \rho \, \sigma^{-1/2}}_{\infty}$, (c) the sandwiched R\'enyi relative entropy ${D}_\alpha(\rho \Vert \sigma) := \frac{1}{\alpha - 1} \ln \left( \operatorname{Tr} \left[ \left( \sigma^{\frac{1 - \alpha}{2\alpha}} \rho \, \sigma^{\frac{1 - \alpha}{2\alpha}} \right)^\alpha \right] \right)$ for $\alpha\in[\frac{1}{2},1)\cup(1,\infty)$, (d) the $\varepsilon$-hypothesis-testing relative entropy $D^{\varepsilon}_{\mathrm{H}}(\rho\Vert\sigma):=-\ln \inf_{0 \leq \Lambda \leq \mathbbm{1}} \left\{ \operatorname{tr}[\Lambda \sigma] \, : \, \operatorname{tr}[\Lambda \rho] \geq 1 - \varepsilon \right\}$, (e) the purified distance $\mathrm{P}(\rho,\sigma):=\sqrt{1-F(\rho,\sigma)}$. $\ln$ denotes natural logarithm. We have $\lim_{\alpha\to 1}D_{\alpha}(\rho\|\sigma)=D(\rho\|\sigma)$ and $\lim_{\alpha\to \infty}D_{\alpha}(\rho\|\sigma)=D_{\infty}(\rho\|\sigma)$. 

The generalized channel divergences are monotonically nonincreasing under the action of quantum superchannels~\cite[Theorem 1]{DGP24}. If the generalized state divergence is nonnegative, then the corresponding channel divergence is also nonnegative. If the generalized state divergence is faithful, i.e., minimum if and only the states are equal, then the corresponding channel divergence is also faithful, i.e., attains minimum if and only if the channels are equal.

An $\varepsilon$-ball around a state $\rho_A$, for $\varepsilon\in[0,1]$, is defined as $\mathcal{B}^\varepsilon(\rho_A)=\{\sigma_A:~ \sigma_A\ge 0, \tr(\sigma_A)\le 1, P(\rho_A,\sigma_A)\le \varepsilon\}$. An $\varepsilon$-ball around a channel $\mc{N}_{A'\to A}$ is defined as $\mathcal{B}^\varepsilon[\mc{N}]=\{\mc{M}\in\Ch(A',A): P[\mc{N},\mc{M}]\leq \varepsilon\}$. For $\rho\in\St(A)$ and $\sigma\in\Pos(A)$, $D_{\infty}^\epsilon(\rho\|\sigma) = \inf_{\omega\in\mc{B}^{\varepsilon}(\rho)} D_{\infty}(\omega\|\sigma)$. $D^\varepsilon_{\infty}(\cdot\Vert\cdot)$ is a generalized state divergence and $D^{\varepsilon}_{\infty}[\cdot\|\cdot]$ is a generalized channel divergence.

The generalized entropy function $\mathbf{S}[\n]$ of a quantum channel $\n_{A'\to A}$ is defined as $\mathbf{S}[\n]:=-\mathbf{D}[\mc{N}\Vert\mc{R}^{\mathbbm{1}}]$~\cite{GW21,SPSD25}, where a completely positive map $\mc{R}^{\omega}_{A'\to A}$ denotes the replacer (or preparation) map that outputs $\omega\in\Pos(A)$, $\mc{R}^{\omega}_{A'\to A}(\rho_{A'})=\tr[\rho_{A'}]\omega_A$ for all $\rho_{A'}\in\Pos(A')$. For a quantum channel $\n_{A'\to A}$ with $|A|<\infty$, $S_{\alpha}[\n]:=\ln|A|-D_{\alpha}[\n\Vert\mc{R}^{\pi}]$~\cite{GW21}, where $\alpha\in[\frac{1}{2},1)\cup (1,\infty)$. We have $\lim_{\alpha\to 1}D_\alpha[\n\Vert\mc{R}^{\pi}]=D[\n\Vert\mc{R}^{\pi}]$. The (von Neumann) entropy of a quantum state $\rho_A$ is $S(\rho_A):=S(A)_{\rho}:=-D(\rho\|\mathbbm{1})=-\tr[\rho\ln\rho]$. $\mc{R}^\pi$ is called uniformly mixing or completely depolarizing channel as it always outputs the maximally mixed state $\pi (\propto \mathbbm{1})$.

\section{Dynamical free energy}\label{sec:freeenergy}
The dynamical free energy, i.e., the free energy of quantum channels, should satisfy certain desirable properties analogous to the free energy of quantum states. It becomes operationally meaningful if it is related to athermality distillation or formation tasks under relevant free operations. It is also preferable that the channel free energy finds operational meaning in a thermodynamic task like work extraction.
  
\textit{Gibbs preservation}.--- For an inverse temperature $\beta\in\mathbbm{R}^+$, the absolutely thermal channel $\mathcal{T}^{\beta}_{A'\to A}$ is a replacer (or preparation) channel that only outputs the thermal state $\gamma^\beta_{A}=\exp(-\beta \widehat{H}_A)/{Z^\beta}$, where $\widehat{H}_A$ denotes the Hamiltonian of $A$ and $Z_A^{\beta}:=\tr[\exp(-\beta \widehat{H}_A)]$ denotes the associated partition function~\cite{Len78}. The thermal state $\gamma^\beta_A$ remains invariant under the translation of the Hamiltonian, $\widehat{H}_A\to \widehat{H}_A+c\mathbbm{1}_A$ for $c\in\mathbbm{R}$. This implies that the absolutely thermal channel $\mc{T}^\beta_{A'\to A}$ also remains invariant under the translation of the output Hamiltonian, $\widehat{H}_A\to \widehat{H}_A+c\mathbbm{1}_A$. We refer to a superchannel $\Theta^{\beta}\in\SCh((A',A),(B',B))$ as Gibbs-subpreserving if $\Theta^{\beta}(\mc{T}^\beta_{A'\to A})\leq \mc{T}^{\beta}_{B'\to B}$, i.e.,
\begin{equation}
 \forall~\psi\in\St(RB'),\quad   \Theta^{\beta}(\mc{T}^\beta_{A'\to A})(\psi_{RB'})\leq \psi_R\otimes\gamma^\beta_B,
\end{equation}
and Gibbs preserving if $\Theta^{\beta}(\mc{T}^\beta_{A'\to A})= \mc{T}^{\beta}_{B'\to B}$. 

For $|A|<\infty$ and $\widehat{H}_A=c\mathbbm{1}_A$, the thermal state is the maximally mixed state, $\gamma^{\beta}=\pi_A$ for all $\beta\in\mathbbm{R}^+$; $\gamma^{\beta}_A=\pi_A$ is true also for arbitrary $\widehat{H}_A$ when $\beta\to 0^+$. The entropy $S[\mathcal{N}]$ of a quantum channel $\mathcal{N}_{A'\to A}$ with the fixed energy, $\sup_{\rho\in\St(A')}\langle \widehat{H}\rangle_{\mathcal{N}(\rho)}=E$, is maximum if and only if the channel is absolutely thermalizing $\mathcal{T}^{\beta}$ with the thermal state of energy $\langle \widehat{H}\rangle_{\gamma^{\beta}}=E$~\cite{Jay57a,PB21,DS25}.

Taking a cue from the discussions in~\cite{DS25,SPSD25,DGP24}, we introduce the notion of the free energy of a quantum channel, which is pivotal to build the framework for a thermodynamic resource theory of quantum channels.
 
 \textit{Thermodynamic and resource-theoretic axioms}.--- Let $f[\mathcal{N}]$ denote a physically motivated free energy function of a quantum channel $\mathcal{N}_{A'\to A}$ in a resource-theoretic paradigm. The function $f:\Ch(A',A)\to \mathbbm{R}^+$ should satisfy the following properties, (A1-A3) essentially and also (A4-A6) preferably, for a fixed inverse temperature $\beta$:
 \begin{enumerate}
     \item[(A1)] Monotonically nonincreasing under the action of a Gibbs preserving superchannel $\Theta^{\beta}$, $\Theta^\beta(\mathcal{T}^{\beta}_{A'\to A})=\mathcal{T}^{\beta}_{B'\to B}$, $f[\Theta^\beta(\mathcal{N})]\leq f[\mathcal{N}]$.
        \item[(A2)] Reduction to the free energy function of the output state for a replacer channel $\mathcal{R}^{\omega}$, $\mathcal{R}^{\omega}(\rho_{A'}):=\omega\in\St(A)$ for all input states $\rho_{A'}$, $f[\mathcal{R}^{\omega}]=f(\omega)$.
     \item[(A3)] Minimum is attained if and only if the channel is absolutely thermal.
      \item[(A4)] Uniform continuity: For any two channels $\mathcal{N},\mathcal{M}$ such that $\frac{1}{2}\norm{\mathcal{N}-\mathcal{M}}_{\diamond}\leq \varepsilon\in[0,1]$, $\abs{f[\mathcal{N}]-f[{\mc{M}}]}\to 0$ as $\varepsilon\to 0$.
     \item[(A5)] Additive under tensor-product channels $\mathcal{N}_{A_1'\to A_1}\otimes\mathcal{M}_{A_2'\to A_2}$ with $\widehat{H}_{A_1A_2}=\widehat{H}_{A_1}\otimes\mathbbm{1}_{A_2}+\mathbbm{1}_{A_1}\otimes\widehat{H}_{A_2}$, $f[\mathcal{N}\otimes\mathcal{M}]=f[\mathcal{N}]+f[\mathcal{M}]$.
     \item[(A6)] Convexity: $f[p\mathcal{N}+(1-p)\mathcal{M}]\leq pf[\mathcal{N}]+(1-p)f[\mathcal{M}]$ for $p\in[0,1]$, where $\mathcal{N}_{A'\to A}$ and $\mathcal{M}_{A'\to A}$ are quantum channels.
 \end{enumerate}
 
The thermalization of a quantum system in contact with a heat bath (surroundings) arises naturally from thermodynamic principles. In a thermodynamic resource theory, it is crucial to fix the inverse temperature, as the thermal state of the quantum system must necessarily have the same inverse temperature as the heat bath for there to be no heat flow between them, ensuring that both the system and the bath are in equilibrium. The thermodynamic utility of a quantum state $\rho_A$ is described by its distinguishability from the thermal state $\gamma^{\beta}_A$, where $\gamma^\beta_A$ represents the thermal state of $A$ in equilibrium with the bath. Resource-theoretic free energy $F^\beta(\rho)$ of the state $\rho_A$ is $F^\beta(\rho)=\beta^{-1}D(\rho\Vert\gamma^{\beta})$~\cite{BHO+13,BHN+15} for $\gamma^{\beta}=\exp(-\beta\widehat{H}_A)/Z^{\beta}_A$. It is related to the thermodynamic non-equilibrium free energy $F^\beta_{\rm T}(\rho)$ in the following way, $\Delta F^\beta_{\rm T}(\rho):= F^\beta_{\rm T}(\rho)-F^\beta_{\rm T}(\gamma^{\beta})=\beta^{-1}D(\rho\|\gamma^{\beta})=F^\beta(\rho)$, where $F^\beta_{\rm T}(\rho)=\langle \widehat{H}\rangle_{\rho}-\beta^{-1}S(\rho)$ for $S(\rho):=-D(\rho\Vert\mathbbm{1})$. Alternatively, $F^\beta_{\rm T}(\rho)=\beta^{-1}D(\rho\|\widehat{\gamma}^{\beta})$, where $\widehat{\gamma}^{\beta}:=\exp(-\beta\widehat{H})$. In resource theory of athermality of quantum states, the generalized free energy function is $\mathbf{F}^\beta(\rho):=\beta^{-1}\mathbf{D}(\rho\Vert\gamma^\beta)$~\cite{BHN+15,Gou24}.
 
 To construct a resource-theoretic framework and coin the dynamical free energy, it is important that the free channel is the absolutely thermal channel with its output in thermal equilibrium with bath. That is, we need to fix the inverse temperature corresponding to the bath. The free energy function of a quantum channel should capture how distinguishable it is from the absolutely thermal channel~\cite{DS25}, and should therefore be monotonically nonincreasing under the actions of Gibbs preserving superchannels. Moreover, the thermodynamic resource theory of quantum states should be recoverable from the thermodynamic resource theory of quantum channels, since quantum states can be seen as preparation channels. For the free energy function of quantum channels to be physically meaningful, it should satisfy essential properties such as continuity and additivity, and preferably convexity (or quasi-convexity), which are well-known for free energy functions of states. The axiomatic properties (A1–A6) are designed to capture these requirements.

\begin{definition}[Generalized free energy]
For a fixed $\beta$, the (resource-theoretic) generalized free energy function $\mathbf{F}^{\beta}[\mathcal{N}]$ of a quantum channel $\mathcal{N}_{A'\to A}$ is defined as the inverse temperature times the generalized divergence of athermality $\mathbf{D}[\n\|\mc{T}^\beta]$ of the quantum channel $\mathcal{N}$,
\begin{equation}\label{eq:f-channel}
    \mathbf{F}^{\beta}[\mc{N}]:={\beta}^{-1}\mathbf{D}[\mathcal{N}\Vert\mathcal{T}^{\beta}],
\end{equation}
$\mathcal{T}^{\beta}_{A'\to A}(\rho_{A'})=\gamma^{\beta}_A$ for all input states $\rho_{A'}$. 
\end{definition} 
The generalized free energy $\mathbf{F}^{\beta}[\mathcal{N}]$ of a quantum channel $\mc{N}_{A'\to A}$ can also be expressed as
\begin{align}\label{eq:f-state}
    \mathbf{F}^{\beta}[\mc{N}] =\frac{1}{\beta}\sup_{\psi\in\St(RA')}\mathbf{D}(\id_R\otimes\mathcal{N}(\psi_{RA'})\Vert\psi_R\otimes\gamma^{\beta}_A),
\end{align}
where it suffices to optimize over pure states $\psi_{RA'}$. The absolutely thermal channel $\mc{T}^\beta_{A'\to A}$ reflects the thermal state of the system $A$ that is in equilibrium with the bath; it entails both, the inverse temperature $\beta$ associated with the bath and the Hamiltonian $\widehat{H}_A$ of the output system.

In the definition of the generalized free energy $\mathbf{F}^{\beta}[\mathcal{N}]$, Eqs.~\eqref{eq:f-channel} and \eqref{eq:f-state}, we could possibly fix the Hamiltonian of the reference $R$ to the channel $\mc{N}$ as per need. For instance, we can always consider $\widehat{H}_R=0$ with the total Hamiltonian of the reference and output to be $\widehat{H}_{RA}=\mathbbm{1}_R\otimes\widehat{H}_A$, if we want the energy of the process $\id_R\otimes\mc{N}$ and $\mc{N}$ to be the same:
\begin{align}
  &  \sup_{\psi\in\St(RA')}\tr[\mathbbm{1}_R\otimes\widehat{H}_A(\id_R\otimes\mc{N}(\psi_{RA'}))] \nonumber\\
  & \qquad \qquad =\sup_{\phi\in\St(A')}\tr[\widehat{H}_A\mc{N}(\phi_{A'})].
\end{align}
See Section~\ref{sec:energy} for a detailed description on the relation between the free energy, entropy, and energy of a quantum channel.

For a quantum channel $\mathcal{N}_{A'\to A}$, its generalized channel divergence with respect to the absolutely thermal channel provides quantification of athermality (nonthermality) in the channel $\n$; we are particularly interested in (a) the free energy $F^{\beta}[\mathcal{N}]:=\beta^{-1}D[\mathcal{N}\Vert\mathcal{T}^\beta]$, (b) the max-free energy $F^{\beta}_{\infty}[\mathcal{N}]:=\beta^{-1}D_{\infty}[\mathcal{N}\Vert\mathcal{T}^{\beta}]$, (c) the sandwiched R\'enyi free energy $F^{\beta}_{\alpha}[\mc{N}]=\beta^{-1}D_{\alpha}[\mathcal{N}\Vert\mathcal{T}^{\beta}]$, (d) the $\varepsilon$-hypothesis-testing free energy ${F}^{\beta,\varepsilon}_{\mathrm{H}}[\mathcal{N}]:=\beta^{-1}D^{\varepsilon}_{\mathrm{H}}[\mc{N}\Vert\mc{T}^{\beta}]$ for $\varepsilon\in[0,1]$, (e) ${F}^{\beta,\varepsilon}_{\infty}[\mathcal{N}]:=\beta^{-1}\inf_{\mc{M}\in\mc{B}^{\varepsilon}[\mc{N}]}D^{\varepsilon}_{\infty}[\mc{M}\Vert\mc{T}^{\beta}]$ for $\varepsilon\in[0,1]$. 

For tensor-product quantum channels $\mathcal{N}_{A'\to A}\otimes\mathcal{M}_{B'\to B}$, we say that the channels $\mc{N}$ and $\mc{M}$ are noninteracting when the total Hamiltonian of the outputs is of the form $\widehat{H}_{AB}=\widehat{H}_A\otimes\mathbbm{1}_B+\mathbbm{1}_A\otimes\widehat{H}_B$, i.e., the interaction Hamiltonian $\widehat{H}^{\rm int}_{AB}=0$. When the total Hamiltonian of the output is $\widehat{H}_{AB}=\widehat{H}_A\otimes\mathbbm{1}_B+\mathbbm{1}_A\otimes\widehat{H}_B$, the absolutely thermal channel $\mathcal{T}^{\beta}_{A'B'\to AB}$ can be written as the tensor-product of local absolutely thermal channels, $\mathcal{T}^{\beta}_{A'B'\to AB}=\mc{T}^{\beta}_{A'\to A}\otimes\mc{T}^\beta_{B'\to B}$.

We observe that the free energy of a quantum channel appears as the maximal conditional free energy associated with the channel, analogous to the channel entropy~\cite{DJKR06,GW21}.
\begin{proposition}\label{prop:free-energy-difference}
The free energy $F^\beta[\n]$ of a quantum channel $\mc{N}_{A'\to A}$ is equal to the maximum channel output free energy conditioned on its reference,
\begin{equation}
F^\beta[\n]=\sup_{\psi\in\St(RA')}\left[F^\beta(\id_R\otimes\mc{N}(\psi_{RA'})) -F^\beta(\psi_{R})\right],
\end{equation}
where it suffices to optimize over pure states $\psi_{RA'}$, $|R|=|A'|$.
\end{proposition}
\begin{proof}
    The proof follows from the definition of the free energy of a quantum channel along with the facts that $F(\rho)=\beta^{-1}D(\rho\Vert\gamma^{\beta})$ and
    \begin{align}
       &\sup_{\psi\in\St(RA')}D(\mc{N}(\psi_{RA'})\Vert\psi_R\otimes\gamma^\beta_A)\nonumber\\
        =&\sup_{\psi\in\St(RA')}\left[D(\mc{N}(\psi_{RA'})\Vert\gamma^\beta_R\otimes\gamma^\beta_A)-D(\psi_R\Vert\gamma^\beta_R)\right],
    \end{align}
    where we used the relation:
    \begin{align}
\tr[\rho_{RA}\ln(\sigma_R\otimes\omega_A)]=\tr[\rho_R\ln\sigma_R]+\tr[\rho_A\ln\omega_A].
    \end{align}
\end{proof}

\subsection{Properties of the dynamical free energy}
All the generalized free energy functions $\mathbf{F}^{\beta}[\mc{N}]$ satisfy (A1) and (A2) readily, whenever $\beta^{-1}\mathbf{D}(\rho\|\gamma^{\beta})$ represents some resource-theoretic free energy function $f(\rho)$ for states $\rho$~\cite{BHN+15,Gou24}. The generalized free energy $\mathbf{F}^{\beta}[\mc{N}]$ remains invariant under the action of Gibbs preserving unitary superchannels. Faithfulness (A3) is satisfied by the generalized free energy whose generalized channel divergence is faithful. For all $\alpha\in[\frac{1}{2},\infty)$, using $\lim_{\alpha\to 1}$ as $\alpha=1$, the sandwiched R\'enyi free energy $F^{\beta}_{\alpha}[\mc{N}]\geq 0$ for an arbitrary quantum channel $\mc{N}_{A'\to A}$, and $F^{\beta}_{\alpha}[\mc{N}]= 0$ iff $\mc{N}=\mc{T}^{\beta}$. For $\alpha\geq 1$ and $n\in\mathbbm{N}$, $F_{\alpha}^{\beta}[\mc{N}^{\otimes n}]=nF^{\beta}_{\alpha}[\mathcal{N}]$ if the outputs of the parallel uses of the channel $\mc{N}$ are noninteracting (i.e., the interaction Hamiltonian between the channel outputs is zero). These assertions follow from \cite[Theorem 1]{DGP24} and \cite[Theorem 4.2]{FGR25}.

\begin{theorem}\label{thm:axioms-satisifaction}
Let $\mc{N}_{A'\to A}$ be an arbitrary quantum channel. The free energy $F^{\beta}[\mc{N}]$ satisfies all the thermodynamical (resource-theoretic) axiomatic properties (A1-A6).
    
For $\alpha\in(1,\infty)$, the sandwiched R\'enyi free energy $F_{\alpha}^{\beta}[\mc{N}]$ satisfies (A1-A5) and quasi-convexity, i.e., for quantum channel $\mc{N}=\sum_x p_x\mc{N}^x$ where $\{p_x\}_x$ is some probability distribution, $F_{\alpha}^{\beta}[\mc{N}]\leq \max_xF_{\alpha}^{\beta}[\mc{N}^x]$. 
\end{theorem}
 We provide a detailed proof in Appendix~\ref{proof:axioms-satisfaction}.

 \begin{lemma} \label{lem:monotonicity_subpreserving}
    The sandwiched R\'enyi free energy $F^{\beta}_{\alpha}[\mc{N}]$, $\alpha\in[\frac{1}{2},\infty)$, is nonincreasing under the action of Gibbs-subpreserving superchannel. It remains invariant under the action of Gibbs preserving unitary superchannels.
\end{lemma}
The proof is given in Appendix \ref{app:proof_lem_monotonicity_subpreserving}.

 We note that the expression for the generalized divergence of athermality $\mathbf{D}[\n\Vert\mc{T}^\beta]$ of a quantum channel $\n_{A'\to A}$ was stated in \cite[Eqs. (67) \& (70)]{SPSD25} in relation to its thermal entropy $\mathbf{S}^\beta[\n]=-\mathbf{D}[\n\Vert\mc{T}^\beta]$. Albeit in a different context, a quantity termed as the robustness of athermality $R^\beta_T[\n]$ of a quantum channel, a particular example of the generalized divergence of athermality of the quantum channel, is studied in~\cite{LMB25}. The robustness of athermality is directly related to the max-relative entropy of athermality (and max-free energy), $D_{\infty}[\n\|\mc{T}^\beta]=\ln(1+R^\beta_T)$ (cf.~\cite{Dat09,LKDW18,LW19,LMB25}).

\begin{lemma}\label{prop:max-free-energy-choi}
    The max-free energy $F^{\beta}_{\infty}[\n]$ of a quantum channel $\n_{A'\to A}$ is equal to
    \begin{equation}\label{eq:max-c-1}
        F^{\beta}_{\infty}[\n]=\beta^{-1}D_{\infty}(\Phi^{\mathcal{N}}_{RA}\Vert\pi_A\otimes\gamma^\beta_A).
    \end{equation}
   For total Hamiltonian of the reference to the channel and its output $\widehat{H}_{RA}=\mathbbm{1}_R\otimes (\widehat{H}_A+c\mathbbm{1}_A)$, where $c\in\mathbbm{R}$, the max-free energy of the channel is equal to the max-free energy of its Choi state (cf.~\cite{LMB25}), 
    \begin{equation}\label{eq:max-c-2}
         F^{\beta}_{\infty}[\n]=F^{\beta}_{\infty}(\Phi^{\mc{N}}_{RA}).
    \end{equation}
\end{lemma}
\begin{proof}
    The proof of Eq.~\eqref{eq:max-c-1} follows from the fact that $D_{\infty}[\mathcal{N}\Vert\mathcal{M}]=D_{\infty}(\Phi^{\mc{N}}_{RA}\Vert\Phi^{\mc{M}}_{RA})$~\cite{LKDW18} and \cite[Eq.~(70)]{SPSD25},
    \begin{align}
        D_{\infty}[\mathcal{N}\Vert\mathcal{T}^{\beta}] & = D_{\infty}({\mc{N}}(\Phi_{RA'})\Vert \mc{T}^{\beta}(\Phi_{RA'}))\\
        & = D_{\infty}(\Phi^{\mathcal{N}}_{RA}\Vert\pi_A\otimes\gamma^\beta_A). 
    \end{align}
    The proof of Eq.~\eqref{eq:max-c-2} as follows from the above. We notice that for $\widehat{H}_{RA}=\mathbbm{1}_R\otimes(\widehat{H}_A+c\mathbbm{1}_A)$ (or, $\widehat{H}_R=c\mathbbm{1}_R$), $\pi_R\otimes\gamma^\beta_A=\gamma^{\beta}_R\otimes\gamma^\beta_A=\gamma^{\beta}_{RA}$. Then~\cite{LMB25},
    \begin{align}
        F^{\beta}_{\infty}[\mc{N}] &=\beta^{-1}D_{\infty}(\Phi^{\n}_{RA}\Vert\gamma^{\beta}_{RA})=F^\beta_{\infty}(\Phi^{\n}_{RA}).
    \end{align}
\end{proof}

\begin{lemma}[\cite{LMB25}]\label{prop:unitary_free_energy}
    The max-free energy $F^{\beta}_{\infty}[\mc{U}]$ of a unitary quantum channel $\mc{U}_{A'\to A}$ is
    \begin{equation}
        F^{\beta}_{\infty}[\mc{U}]= \beta^{-1}\ln \tr\left[{(\gamma^{\beta}_A)}^{-1}\right].
    \end{equation}
    If the Hamiltonian of $A$ is fully degenerate, $\widehat{H}_A\propto \mathbbm{1}_A$, then $F^{\beta}_{\infty}[\mc{U}]=2\beta^{-1}\ln|A|$.
\end{lemma}
See \cite[Lemma 6]{BD26} for the max-free energy of an isometry quantum channel. For the sake of completeness, the proof of Lemma~\ref{prop:unitary_free_energy} is given in Appendix~\ref{proof:prop:unitary_free_energy}.

\begin{proposition}\label{thm:unitary-free-energy}
    For $\alpha\in\{1,\infty\}$ and $|A'|=|A|$, the sandwiched R\'enyi free energy $F^{\beta}_{\alpha}[\mc{N}]$ of a quantum channel $\n_{A'\to A}$ is maximum if and only if $\n$ is a unitary channel. For all unitary channels $\mc{U}_{A'\to A}$, $F^{\beta}_{\alpha}[\mc{U}]=F^{\beta}_{\alpha}[\id_{A'\to A}]$.
\end{proposition}
    Details of the proof are in Appendix~\ref{app:proof-thm-unitary-free-energy}; the proof for $\alpha=\infty$ follows from \cite[Corollary A.4.2]{LMB25}, since $F_{\infty}^\beta[\n]=\beta^{-1}\ln(1+R^\beta_T[\n])$.

\begin{lemma}\label{lem:free-ordering}
   Let $\mc{N}$ be a quantum channel. The generalized free energies obey order $\mathbf{F}^{\beta,(x)}[\mc{N}]\leq \mathbf{F}^{\beta,(y)}[\mc{N}]$ if respective generalized state divergence obey the same order, i.e., $\mathbf{D}^{(x)}(\rho\Vert\sigma)\leq \mathbf{D}^{(y)}(\rho\Vert\sigma)$ for quantum states $\rho,\sigma$.
\end{lemma}
\begin{proof}
    The proof follows from the definition of the generalized free energy and proof argument similar to \cite[Lemma 1]{DGP24}.
\end{proof}

A direct consequence of the above lemma is that, for all $\alpha_1, \alpha_2\in[\frac{1}{2},\infty)$ such that $\alpha_1\leq \alpha_2$, we have $F^{\beta}_{\alpha_1}[\n]\leq F^{\beta}_{\alpha_2}[\n]$. 

\begin{lemma}
    The free energy $F^\beta[\n]$ of a quantum channel $\n_{A'\to A}$ is lower bounded by the completely-bounded norm of the channel with the absolutely thermal channel $\mc{T}^\beta_{A'\to A}$,
    \begin{equation}
        F^\beta[\n] \geq \frac{1}{2\beta}\norm{\n-\mc{T}^\beta}_{\diamond}^2.
    \end{equation}
\end{lemma}
\begin{proof}
    The proof follows from Pinsker's inequality for any two states $\rho,\sigma$:
    \begin{equation}
        D(\rho\Vert\sigma)\geq \frac{1}{2}\norm{\rho-\sigma}_1^2.
    \end{equation}
    The Pinsker's inequality implies that, for any two quantum channels $\n,\mc{M}\in\Ch(A',A)$,
    \begin{equation}
        D[\n\Vert\mc{M}]\geq \frac{1}{2}\norm{\n-\mc{M}}_{\diamond}^2.
    \end{equation}
\end{proof}

\section{Dynamical Resource Theory of Athermality}\label{sec:resourcetheory}
We now introduce framework and characterization of the thermodynamic resource theory of square quantum channels. A quantum channel $\n_{A'\to A}$ is a square quantum channel if $|A'|=|A|$; in other words, a quantum channel with square Choi operator is called a square channel. To quantify the thermodynamic resourcefulness of a quantum channel $\n$, we consider its distinguishability from a thermal channel $\mc{T}^{\beta}$ whose output $\gamma^\beta$ is in thermal equilibrium with the bath~\cite{DS25}. The absolutely thermal channel $\mc{T}^{\beta}$ represents thermalization process with respect to the given bath at inverse temperature $\beta$. 

To characterize a resource theory of athermality of quantum channels, we identify our free object to be the absolutely thermal channel $\mc{T}^\beta$ and the free operations to be Gibbs preserving superchannels (GPSC). This is a relevant assumption as thermalization processes seem to be natural in practical scenarios where quantum systems inherently tend to thermalize upon interaction with the bath. The closed systems in isolation with the bath are nontrivial and idealistic. We write the athermal channel resource as $(\n,\mathcal{T}^\beta)$ to specify that the athermality of the channel $\n$ is being considered with respect to $\mc{T}^{\beta}$. We call two channel-resources $(\n,\mathcal{T}^\beta)$ and $(\mathcal{M},\mathcal{T}^{\beta'})$ to be equivalent, $(\n,\mathcal{T}^\beta)\sim(\mathcal{M},\mathcal{T}^{\beta'})$, if they are interconvertible through some GPSCs.

For a quantum superchannel $\Theta\in\SCh((A',A),(B',B))$ with the following action,
\begin{equation}
    \Theta(\n_{A'\to A})=\mc{Q}_{A\to B}\circ\n\circ\mc{P}_{B'\to A'},
\end{equation}
$\Theta$ is Gibbs preserving if and only if $\mc{Q}_{A\to B}$ is Gibbs preserving channel, i.e., $\mc{Q}(\gamma^\beta_A)=\gamma^{\beta}_B$. A quantum superchannel $\widetilde{\Theta}\in\SCh((A',A),(B',B))$, where $\widetilde{\Theta}(\n)=\widetilde{\mc{Q}}_{PA\to B}\circ\n\circ\widetilde{\mc{P}}_{B'\to PA'}$ is also Gibbs preserving if its postprocessing channel $\widetilde{\mc{Q}}_{PA\to B}$ is absolutely thermal. 

 \textit{Golden units of athermality}.--- For a square quantum channel $\n_{A'\to A}$, the free energy is maximum if and only if the channel is unitary. For an arbitrary pair $\mc{U}_{A'\to A},\mc{V}_{A'\to A}$ of unitary channels, there always exists a GPSC $\Theta^{\beta}\in\SCh((A',A),(A',A))$ such that $\Theta^{\beta}(\mc{U})=\mc{V}$, for an instance, $\Theta^{\beta}(\n)=\n\circ(\mc{U}^{\dag}\circ\mc{V})$. That is, all unitary channels $(\mc{U},\mc{T}^\beta)$ are equally resourceful, $(\id_{A'\to A},\mc{T}^\beta_{A'\to A})\sim(\mc{U}_{A'\to A},\mc{T}^\beta_{A'\to A})$. This is consistent with the free energy of all unitary channels $(\mc{U},\mc{T}^\beta)$ being the same. The free energy of a unitary channel is least when the output Hamiltonian is fully degenerate (see Lemma~\ref{lemma:free-max-energy}). With freedom of choice for golden units and the fact that GPSCs are free operations in our athermality framework, we mark our golden unit as  the identity resource channel $(\id_{A'\to A},\mathcal{R}^{\pi})$, where $\mc{R}^{\pi}_{A'\to A}$ is the absolutely thermal channel for fully degenerate output Hamiltonian ($\widehat{H}_A\propto \mathbbm{1}_A$). We denote $(\id_{A'\to A},\mc{R}^\pi_{A'\to A})$ as $(\id_m,\mc{R}^\pi)$ for $|A'|=|A|=m$.

 It is evident that $(\id_{A'\to A},\mathcal{R}^{\pi}_{A'\to A})$ is more resourceful than $(\id_{B'\to B},\mathcal{R}^{\pi}_{B'\to B})$ if $|A|>|B|$~\cite{GW21,DS25,SPSD25}, therefore, we will have to be considerate of the dimension of $\id$ in the golden unit. The identity channel $\id_m$ (or any equivalent unitary channel $\mc{U}_m$) allows for the perfect distribution of $m\times m$-dimensional quantum states between two parties~\cite{Sch96,ADHW09}, equivalent to the distribution of a maximally entangled state $\Phi$ of Schmidt rank $m$~\cite{DW05,ADHW09,KDWW19}.

The uniformly mixing channel $\mathcal{R}^\pi_{A'\to A}$, $|A'|=|A|=m$, can be obtained by the action of unitary channels $\mc{U}(\cdot)=U(\cdot){U}^\dag$ uniformly picked at random with respect to Haar measure,
\begin{equation}
    \mathcal{R}^\pi(\cdot)=\int_{U\in\mathbbm{U}(m)}{\operatorname{d}\!\mu_U}U(\cdot)U^\dag,
\end{equation}
where $\mathbbm{U}(m)$ is the set of all unitary operators $U_{A'\to A}$. We can also obtain $\mathcal{R}^\pi$ by taking a uniform mixture of Weyl unitary channels $\mathcal{W}^i(\cdot)=W^i(\cdot){W^i}^\dagger$,
\begin{align}
    \mathcal{R}^\pi(\cdot) &=\frac{1}{m^2}\sum_{i=0}^{m^2-1} {W}^i(\cdot){W^i}^\dag,
\end{align}
where $\{W^i\}_{i=0}^{m^2-1}$ is the set (group) of Weyl unitaries; $\{W^i\}_{i=0}^{m^2-1}$ forms a complete orthonormal basis for the space of all linear operators acting on $m$-dimensional Hilbert space. Let $\mc{W}^0=\mathrm{\id}_m$ and $\id_m^{\perp}:=\sum_{i=1}^{m^2-1}\mathcal{W}^i$, then
\begin{align}
 \mc{R}^\pi  & =\frac{1}{m^2}\id_m+\frac{1}{m^2}\sum_{i=1}^{m^2-1}\mathcal{W}^i=\frac{1}{m^2}(\id_m+\id_m^{\perp}).
\end{align}
This implies~(cf.~\cite[Proposition 11]{DKSW18})
\begin{equation}
    \frac{1}{2}\norm{\id_m-\mc{R}^\pi}_{\diamond}=\frac{1}{2}\norm{\mc{U}-\mc{R}^\pi}_{\diamond}=1-\frac{1}{m^2}.
\end{equation}

\textit{Channel conversion distance}.-- For a given resource channel $(\n,\mathcal{T}^{\beta}_1)$ being transformed to another resource channel $(\mc{M},\mathcal{T}^{\beta}_2)$, the conversion distance under GPSCs are defined as
    \begin{align}
      &  d_{\mathrm{GP}}((\n,\mathcal{T}^{\beta}_1)\rightarrow(\mc{M},\mathcal{T}^{\beta}_2))\nonumber\\
        &\qquad=\min_{\Theta}\left\{P[\mc{M},\Theta(\n)]~:~\Theta(\mathcal{T}^{\beta}_1)=\mathcal{T}^{\beta}_2\right\}.
    \end{align}
Here $P[\mc{M},\mc{K}]$ is the purified channel distance between the channels $\mc{M}$ and $\mc{K}$ given by
\begin{align}
    P[\mc{M},\mc{K}]:=\sup_{\psi\in \St(RA')} P(\mc{M}(\psi_{RA'}),\mc{K}(\psi_{RA'}))
\end{align}
where $P(·,·)$ is the purified distance, and it suffices to take $\psi_{RA'}$ to be pure.
\begin{definition}[One-shot athermality distillation and formation]
For an error $\varepsilon\in[0,1]$, the single-shot athermality distillation of a resource channel $(\n,\mc{T}^\beta)$ under GPSCs is defined as
\begin{align}
        \mathrm{Dist}&^\varepsilon (\n,\mathcal{T}^\beta)\nonumber \\
        &:=\sup_{m}\left\{\ln m:~d_{\mathrm{GP}}((\n,\mathcal{T}^\beta)\rightarrow(\id_m,\mathcal{R}^\pi))\le\varepsilon\right\}.
    \end{align}
For an error $\varepsilon\in[0,1]$, the single-shot athermality formation (cost) of a resource channel $(\n,\mathcal{T}^\beta)$ is defined as 
    \begin{align}
        \mathrm{Cost}&^\varepsilon (\n,\mathcal{T}^\beta)\nonumber \\
        &:=\inf_{m}\left\{\ln m:~    d_{\mathrm{GP}}((\id_m,\mathcal{R}^\pi)\rightarrow(\n,\mathcal{T}^\beta))\le\varepsilon\right\}.
    \end{align}
\end{definition}

The single-shot athermality distillation $\mathrm{Dist}^\varepsilon (\n,\mathcal{T}^\beta)$ provides the largest dimensional golden unit $(\id_m,\mc{R}^{\pi})$ that can be distilled from the resource channel $(\n,\mathcal{T}^\beta)$ under GPSCs, up to allowed error $\varepsilon$. We could possibly also interpret $(\ln 2)^{-1}\mathrm{Dist}^\varepsilon (\n,\mathcal{T}^\beta)$ as the maximum number of elementary golden units $(\id_2,\mc{R}^{\pi})$ distillable from the resource channel $(\n,\mathcal{T}^\beta)$ under GPSCs. The single-shot athermality formation $\mathrm{Cost}^\varepsilon (\n,\mathcal{T}^\beta)$ provides the least dimensional golden unit $(\id_m,\mc{R}^{\pi})$ needed to form the resource channel $(\n,\mathcal{T}^\beta)$ under GPSCs, up to allowed error $\varepsilon$. We could possibly also interpret $(\ln 2)^{-1}\mathrm{Cost}^\varepsilon (\n,\mathcal{T}^\beta)$ as the minimum number of elementary golden units $(\id_2,\mc{R}^{\pi})$ spent to form the resource channel $(\n,\mathcal{T}^\beta)$ under GPSCs.

\begin{theorem}[Athermality distillation and formation]\label{thm:dist_cost}
For any error $\varepsilon\in[0,1]$ and a given resource channel $(\n,\mathcal{T}^\beta)$, the single-shot athermality distillation and formation are proportional to the $\varepsilon$-hypothesis-testing free energy and the $\varepsilon$-max-free energy of the channel $\n$, respectively,
\begin{align}
    \mathrm{Dist}^\varepsilon (\n,\mathcal{T}^\beta)&=\frac{1}{2}D_{H}^{\varepsilon^2}[\n\Vert\mathcal{T}^\beta],\\
     \mathrm{Cost}^\varepsilon (\n,\mathcal{T}^\beta)&=\frac{1}{2}D_{\infty}^\varepsilon[\n\Vert\mathcal{T}^\beta],
\end{align}
$F^{\beta,\varepsilon}_H[\n]=\frac{2}{\beta}\mathrm{Dist}^{\sqrt{\varepsilon}} (\n,\mathcal{T}^\beta)$ and $F^{\beta,\varepsilon}_{\infty}[\n]=\frac{2}{\beta}\mathrm{Cost}^\varepsilon (\n,\mathcal{T}^\beta)$.
\end{theorem}
We provide the detailed proof of the theorem in Appendix~\ref{app:proof_thm_dist_cost}. The zero-error cost is proportional to the channel max-free energy $F_\infty^\beta[\n]$, that can be expressed as a semidefinite program, see Appendix~\ref{app:sdp}.

The one-shot athermality distillation and formation of a resource channel $(\n,\mc{T}^\beta)$ are associated with how well the channel $\n$ can be discriminated from $\mc{T}^\beta$ in an asymmetric channel discrimination task~\cite{CMW16,DW19b,WBHK20,BKSD23}, where $\mc{N}$ is null hypothesis and $\mc{T}^\beta$ is alternate hypothesis. Broadly speaking, there are two different strategies one can apply when multiple uses of channels are allowed, namely, adaptive and nonadaptive~\cite{CMW16,BKSD23}. Channel discrimination under nonadaptive strategy can be thought of as a special case of discrimination task with adaptive strategy. Nonadaptive strategy allows for parallel uses of a channel with free preprocessing and postprocessing operations, while adaptive strategy allows for sequential uses of channel with successive uses of channel being interleaved with a free adaptive channel~\cite{CMW16,DW19b}.

It intuitively holds that the asymptotic resource distillation rate is lesser than or equal to the formation rates in a resource theory where golden units (most resourceful objects) are distilled from lesser resource objects and spent as costs to form lesser resource objects (cf.~\cite{HSDW24}). It follows from the relation between the min relative entropy and max relative entropy, $D^0_{\rm H}(\rho\|\sigma)\leq D_{\infty}(\rho\|\sigma)$, that the zero-error athermality distillation and formation satisfy $\mathrm{Dist}^{0} (\n,\mathcal{T}^\beta)\leq \mathrm{Cost}^{0} (\n,\mathcal{T}^\beta)$.

Based on the assistance of adaptive and nonadaptive strategies for the channel discrimination tasks (see~\cite{CMW16,FFRS20,BKSD23} for formal discussion), for an error $\varepsilon\in(0,1)$, we define the adaptive (${\rm ad}$) and nonadaptive ($\parallel$) asymptotic athermality distillation and formation rates of a resource channel $(\n,\mc{T}^\beta)$ as
\begin{align}
   \mathscr{C}_{\mathrm{distill}}^{{\varepsilon,\rm ad}}[\mathcal{N}] 
   &:= \limsup_{n \to \infty} \frac{1}{n} \, \mathrm{Distill}^{\varepsilon}_{\rm ad} 
   \left( \mathcal{N}^{(n)}, \left( \mathcal{T}^\beta \right)^{(n)} \right),\\
    \mathscr{C}^{\varepsilon,{\rm ad}}_{\mathrm{cost}}[\mathcal{N}] 
   &:= \liminf_{n \to \infty} \frac{1}{n} \, \mathrm{Cost}^{\varepsilon}_{\rm ad} 
   \left( \mathcal{N}^{(n)}, \left( \mathcal{T}^\beta \right)^{(n)} \right),\\
    \mathscr{C}_{\mathrm{distill}}^{\varepsilon,\parallel}[\mathcal{N}] 
   &:= \limsup_{n \to \infty} \frac{1}{n} \, \mathrm{Distill}^{\varepsilon}_{\parallel} 
   \left( \mathcal{N}^{\otimes n}, \left( \mathcal{T}^\beta \right)^{\otimes n} \right),\\
    \mathscr{C}^{\varepsilon,\parallel}_{\mathrm{cost}}[\mathcal{N}] 
   &:= \liminf_{n \to \infty} \frac{1}{n} \, \mathrm{Cost}^{\varepsilon}_{\parallel} 
   \left( \mathcal{N}^{\otimes n}, \left( \mathcal{T}^\beta \right)^{\otimes n} \right).
\end{align}
It directly follows from the definitions that $\mathscr{C}_{\mathrm{distill}}^{\varepsilon,\parallel}[\mathcal{N}] \leq \mathscr{C}_{\mathrm{distill}}^{\varepsilon,{\rm ad}}[\mathcal{N}]$ and $ \mathscr{C}^{\varepsilon,\parallel}_{\mathrm{cost}}[\mathcal{N}] \leq\mathscr{C}^{\varepsilon,{\rm ad}}_{\mathrm{cost}}[\mathcal{N}]$, and $\mathscr{C}_{\mathrm{distill}}^{\varepsilon,\parallel}[\mathcal{N}]\leq \mathscr{C}^{\varepsilon,\parallel}_{\mathrm{cost}}[\mathcal{N}] \leq\mathscr{C}^{\varepsilon,{\rm ad}}_{\mathrm{cost}}[\mathcal{N}]$. We prove that the asymptotic distillation and formation rates are equal under nonadaptive strategies. 

\begin{theorem}[Asymptotic reversibility]\label{thm:rev}
The resource theory of athermality of quantum channels with GPSCs as free operations is asymptotically reversible under parallel uses of channels, i.e., for any $\varepsilon\in(0,1)$ and an arbitrary (square) quantum resource channel $(\n,\mc{T}^\beta)$,
  \begin{align}
       \mathscr{C}_{\mathrm{distill}}^{\varepsilon,\parallel}[\mathcal{N}]=\mathscr{C}_{\mathrm{distill}}^{\varepsilon,{\rm ad}}[\mathcal{N}]=\mathscr{C}_{\mathrm{cost}}^{\varepsilon,\parallel}[\mathcal{N}]&=\frac{1}{2}D[\n\Vert\mc{T}^\beta]\nonumber\\
       &=\frac{\beta}{2} F^\beta[\n].
  \end{align} 
\end{theorem}
\begin{proof}
    For any $\varepsilon\in(0,1)$, we have from~\cite[Theorem 1]{CMW16},
    \begin{equation}
        \mathscr{C}_{\mathrm{distill}}^{\varepsilon,\parallel}[\mathcal{N}]=\mathscr{C}_{\mathrm{distill}}^{\varepsilon,{\rm ad}}[\mathcal{N}]=\frac{1}{2} D[\n\Vert\mc{T}^\beta],
    \end{equation}
  and from~\cite[Theorem 4.1]{FGR25},
    \begin{equation}
        \mathscr{C}_{\mathrm{cost}}^{\varepsilon,\parallel}[\mathcal{N}]=\frac{1}{2}D[\n\Vert\mc{T}^\beta].
    \end{equation}
    We conclude the proof by recalling that $F^\beta[\n]=\beta^{-1}D[\n\|\mc{T}^\beta]$.
\end{proof}
It follows from \cite[Theorem 2]{CMW16} that for any athermality distillation scheme which seeks to distill at a rate
strictly higher than $\frac{1}{2}D[\n\|\mc{T}^\beta]$, the probability of successful distillation decays to zero exponentially fast with the number of channel uses of 
$\n$. That is, $\frac{1}{2}D[\n\|\mc{T}^\beta]$ also serves the strong converse bound on the asymptotic athermality distillation rate of the channel $\n$, irrespective of whether the strategy used is adaptive or nonadaptive.

Under parallel uses of channels, both the athermality distillation and formation rates of a quantum resource channel $(\n,\mc{T}^\beta)$ converge to the same value, which is half the relative entropy between the channel and the absolutely thermal channel $\mc{T}^\beta$. Therefore, the resource theory of athermality under GPSCs is reversible. The asymptotic athermality distillation and formation rates of a tensor-product quantum resource channel $(\n\otimes\mc{M},\mc{T}^\beta)$, where the channel outputs of quantum channels $\mc{N}_{A'\to A}$ and $\mc{M}_{B'\to B}$ are noninteracting and $\mc{T}^\beta_{A'B'\to AB}=\mc{T}^\beta_{A'\to A}\otimes\mc{T}^\beta_{B'\to B}$, are equal to $\frac{\beta}{2}\left(F^\beta[\n]+F^\beta[\mc{M}]\right)$. This is a straightforward implication of the additivity of the free energy under tensor-product channels with noninteracting channel outputs.

\subsection{Remarks on the resource theory of athermality}
We note that in \cite{NG15}, the restricted form of the free energy associated to athermal channel, where reference to the channel are not considered, is shown to be related to the work distillation. The allowed free operations for their distillation tasks are thermal operations, a strict subset of Gibbs preserving channels. They showed that for the asymptotically many uses of each channel, the distillable work is an additive function of the considered channels. In \cite{NG15,FBB21}, the identity channel has zero distillable work which is in stark contrast with our framework where the identity channel is one of the most resourceful athermal channel. We take an alternate approach where access to a reference to the channel is allowed along with adaptive strategies for the uses of the channel in the distillation tasks. The access to the channel reference and adaptive strategy for distillation tasks provide a general approach in resource-theoretic framework that involves channel discrimination. The rates of distillation will never be less than that of the restricted case with parallel uses of channels and no allowance of the reference to the channel. The sandwiched R\'enyi free energy is additive for tensor-product channels with noninteracting outputs for $\alpha\in[\frac{1}{2},\infty)$, including the free energy ($\alpha=1$). This additivity property of the free energy makes the asymptotic rates of distillation additive in our resource theory of athermality. 

Our dynamical resource theory of athermality generalizes the static thermodynamic resource theory, i.e., resource theory of athermality of states, under Gibbs preserving channels as free operations. The one-shot and asymptotic distillation and formation (costs) of a quantum state under Gibbs preserving channels~\cite{Gou24} can be obtained from our framework by considering preparation (or replacer) channel which always outputs the given state. We have to be cognizant (careful) that a golden unit in static thermodynamic resource theory is $(\op{m},\pi)$~\cite{Gou24}, where $\op{m}_A$ is a pure energy eigenstate, and all energy eigenstates $\op{i}$ are equally probable when $\widehat{H}_A=0$, $|A|=m$. In regard with the dynamical resource theory, a golden unit is $(\id_m,\mc{R}^\pi)$, we allow for reference in the case of channel which effectively means that we are dealing with quantum systems of size $|A'A|=|A'||A|=m^2$. The Choi state of $\mc{R}^\pi$ is the maximally mixed state $\pi_{A'}\otimes\pi_A=\pi_{A'A}$ and the Choi state of $\id$ is a maximally entangled state $\Phi_{A'A}$, both are density operators on $m^2$-dimensional Hilbert space. We may say that the dynamical resource $(\id_m,\mc{R}^\pi)$ is equivalent to static resource $(\Phi_m,\pi_{A'}\otimes\pi_A)$, where $\Phi_m$ is a maximally entangled state of Schmidt rank $m$. To calculate the asymptotic distillation and formation rates of a given state, they are respectively going to be twice of the asymptotic distillation and formation rates of the replacer channel that outputs the given state.

The dynamical thermodynamic resource theory exhibit some features beyond the static thermodynamic resource theory. When Hamiltonian $\widehat{H}_A$ is non-degenarate, i.e., no two distinct energy eigenstates have same eigenvalues, then a state has the maximum free energy if and only if it is the most excited energy eigenstate (eigenstate with the highest energy eigenvalue). The state with maximum free energy for a quantum system with non-degenerate Hamiltonian is unique. This is characteristically different from the dynamical case where all the unitary channels $\mc{U}_{A'\to A}$ have maximum free energy. We do have the resource theory of entanglement of states where there are many golden units called maximally entangled states~\cite{HHHH09}, which are all equivalent to each other up to the action of local unitary channels.

The conversion of a resource channel $(\mc{N}_{A'\to A},\mc{T}^\beta_{A'\to A})$ to $(\mc{T}^\beta_{A'\to A},\mc{T}^\beta_{A'\to A})$ is an informational thermodynamic task of erasing the channel $\n$ via thermalization~\cite{Ben03,MDP22,BGC+25,JGW25}, whereas the conversion of the absolutely thermal channel $(\mc{T}^{\beta},\mc{T}^\beta)$ to $(\n,\mc{T}^{\beta})$ is an informational thermodynamic task of preparing or forming the channel.

\section{Informational aspects of free energy}\label{sec:aspects}
We discuss connection of the dynamical free energy with the channel capacity and the free energy of the channel output. We observe connections of the resource theory of athermality with the information processing tasks of private randomness and purity distillation.

The quantum mutual information $I(A;B)_{\rho}$ of a quantum state $\rho_{AB}$ is defined as
\begin{align}
    I(A;B)_{\rho}& =D(\rho_{AB}\Vert\rho_A\otimes\rho_B)\nonumber\\
    &=S(A)_{\rho}+S(B)_{\rho}-S(AB)_{\rho}.
\end{align}
It is the minimal rate of noise needed to erase the total correlations in the state $\rho_{AB}$ so that the result is a product state $\rho_A\otimes\rho_B$~\cite{GPW05}.

The quantum mutual information $I(\n,\rho)$ of a quantum channel $\n_{A'\to A}$ and an input state $\rho_{RA'}$ to $\id_R\otimes\n$ is defined as the quantum mutual information $I(R;A)_{\n(\rho_{RA'})}$, $I(\n,\rho):=I(R,A)_{\n(\rho_{RA'})}=D(\n(\rho_{RA'})\Vert\rho_R\otimes\n(\rho_{A'}))$. The quantum mutual information of a quantum channel $\n_{A'\to A}$ is defined as
\begin{equation}
I[\n]:=\sup_{\psi\in\St(RA')}I(R;A)_{\n(\psi_{RA'})},    
\end{equation}
where it suffices to optimize over pure state $\psi_{RA'}$. The quantum mutual information of a channel indicates how well the channel can preserve the total correlation (quantum+classical) between the reference and the input. It finds operational meaning as the entanglement-assisted classical capacity of the channel~\cite{BSST02} (see also~\cite{DW19,DW19b} for connection with quantum rebound capacity) and the asymptotic quantum simulation cost under no-signalling assisted codes~\cite{FWT20}.

\begin{proposition}\label{prop:mi}
    The free energy $F^\beta[\n]$ of a quantum channel $\n_{A'\to A}$ satisfies
    \begin{equation}
        F^\beta[\n] = \sup_{\psi\in\St(RA')}\left[\beta^{-1} I(R;A)_{\mc{N}(\psi)}+F^\beta(\n(\psi_{A'}))\right],
    \end{equation}
    where it suffices to optimize over pure states $\psi_{RA'}$, $|R|=|A'|$.
\end{proposition}
\begin{proof}
    The proof follows from the following observations. $F[\n]=\beta^{-1}\sup_{\psi\in\St(RA')}D(\n(\psi_{RA'})\Vert\psi_R\otimes\gamma^\beta_A)$, where it suffices to optimize over pure states $\psi_{RA'}$, and
\begin{align}
    &   \beta^{-1}\sup_{\psi\in\St}D(\n(\psi_{RA'})\Vert\psi_R\otimes\gamma^\beta_A)\nonumber\\
    = &\beta^{-1}\sup_{\psi\in\St(RA')}\left[D(\psi^{\n}_{RA}\Vert\psi_R\otimes\psi^{\n}_A)+D(\psi^{\n}_A\Vert\gamma^\beta_A) \right]\nonumber\\
    =& \sup_{\psi\in\St(RA')}\left[\beta^{-1}I(R;A)_{\n(\psi_{RA'})}+ F(\n(\psi_{A'})) \right],
\end{align}
where $\psi^{\n}_{RA}:=\n(\psi_{RA'})$.
\end{proof}

The corollary given below follows directly from the above proposition.
\begin{corollary}\label{cor:lower-bound}
    The free energy ${F}^{\beta}[\n]$ of a quantum channel $\n_{A'\to A}$ is lower bounded as
    \begin{equation}
        F^{\beta}[\n]\geq \beta^{-1}I(R;A)_{\Phi^{\n}}+F(\n(\pi_{A'})),
    \end{equation}
    and upper bounded as
    \begin{equation}
        F^{\beta}[\n]\leq \beta^{-1}I[\n]+\sup_{\psi\in\St(A')}F^{\beta}(\mc{N}(\psi)).
    \end{equation}
    The free energy of a unitary channel $\mc{U}_{A'\to A}$ is lower bounded as
    \begin{equation}
        F^\beta[\mc{U}]\geq \beta^{-1}\left(\ln|A|+\ln Z^{\beta}_A\right)+\frac{1}{|A|}\tr[\widehat{H}_A],
    \end{equation}
    where $Z^{\beta}_A:= \tr[\exp(-\beta\widehat{H}_A)]$.
\end{corollary}

\begin{proposition}
Given a unitary channel $\mc{U}_{A'\to A}$, the minimum generalized free energy of the channel is achieved for the output system Hamiltonian $\widehat{H}_A=c\mathbbm{1}_A$, In particular, for $\alpha\in[1,\infty)$,
    \begin{align}
\min_{\widehat{H}_A}F^\beta_\alpha[\mc{U}]=2\beta^{-1}\ln|A|,~\forall~\alpha\in [1,\infty).
    \end{align}
\end{proposition}
\begin{proof}
From Corollary~\ref{cor:lower-bound}, we have $F^\beta[\mc{U}]\ge f(\widehat{H}_A)$ where 
\begin{align}
f(\widehat{H}_A)=\beta^{-1}(\ln|A|+\ln Z_A^\beta+\frac{\beta}{|A|}\tr(\widehat{H}_A)).
\end{align}
To find the global minima of the above operator function, we analyze the operator derivative $\nabla_{\widehat{H}_A}f(\widehat{H}_A)$. Using arguments similar to Lemma~\ref{lemma:free-max-energy} in Appendix~\ref{app:lemma-free-max-energy}, we can show that
\begin{align}
\min_{\widehat{H}_A}f(\widehat{H}_A)=f(c\mathbbm{1}_A)=2\beta^{-1}\ln|A|.
\end{align}
Combining the above result with that of Lemma~\ref{lemma:free-max-energy} and Lemma~\ref{lem:free-ordering}, we have the following for $\alpha\in (1,\infty)$, 
\begin{align}
    2\beta^{-1}\ln|A|\le&\min_{\widehat{H}_A} F^\beta[\mc{U}]\le \min_{\widehat{H}_A}F_\alpha^\beta[\mc{U}]\\\nonumber
    \le& \min_{\widehat{H}_A}F_\infty^\beta[\mc{U}]=  2\beta^{-1}\ln|A|\\
  \implies  \min_{\widehat{H}_A} F^\beta[\mc{U}]= &\min_{\widehat{H}_A}F_\alpha^\beta[\mc{U}]= \min_{\widehat{H}_A}F_\infty^\beta[\mc{U}]=  2\beta^{-1}\ln|A|.
\end{align}
\end{proof}
\textit{Private randomness distillation}.---  The private randomness capacity of a channel $\n_{A'\to A}$ is the maximum rate, in an asymptotic setting of i.i.d. uses of the channel, at which the receiver accessing output $A$ can extract private randomness, against an eavesdropper accessing extension $E$ of an isometric channel extension $\mc{V}^{\n}_{A'\to AE}$ of the channel $\n$, when the sender sends states through the channel~\cite{YHW19}. The private randomness capacity $P[\n]=D[\n\|\mc{R}^\pi]=\lim_{\beta\to 0^+}D[\n\|\mc{T}^\beta]$. The private randomness capacity of a square quantum channel is twice its asymptotic thermodynamic distillation or formation rate under the action of GPSCs in the limit inverse temperature tends to vanish.

\textit{Dynamical resource theory of purity}.---  By considering $\beta\to 0^+$ or the Hamiltonian of the channel output to be proportional to $\mathbbm{1}$ in our dynamical resource theory of athermality, we obtain the dynamical resource theory of purity~\cite{YHW19,LY20,RT21,YZGZ20} where unitary channels are golden units and the uniformly mixing channels are free objects. When $\widehat{H}_A\propto \mathbbm{1}_A$, then $\mc{T}^\beta_{A'\to A}=\mc{R}^{\pi}_{A'\to A}$ as $\gamma^\beta_A=\pi_A$ due to $\widehat{H}_A$ being fully degenerate Hamiltonian, i.e., all energy eigenstates have the same eigenvalue. The asymptotic purity distillation and formation rate of a square quantum channel $\mc{N}_{A'\to A}$ is equal to $\frac{1}{2}D[\n\|\mc{R}^\pi]=\frac{1}{2}\lim_{\beta\to 0^+}D[\n\|\mc{T}^\beta]$~\cite{YHW19,GW21,YZGZ20}. 

\textit{Entropy and thermal entropy}.--- The entropy $S[\n]$ of a quantum channel $\n_{A'\to A}$ and its private randomness capacity $P[\n]$ satisfy the following duality or trade-off relation~\cite{DGP24,DS25}:
\begin{equation}
    S[\n]+P[\n]=\ln|A|.
\end{equation}
We have a similar dual relationship between the thermal entropy $S^{\beta}[\n]:=\ln Z^{\beta}_A-D[\n\|\mathcal{T}^{\beta}]$~\cite{SPSD25} and the resource-theoretic free energy $F^{\beta}[\n]$ of a quantum channel $\n_{A'\to A}$:
\begin{equation}
    S^{\beta}[\n]+\beta F^{\beta}[\n]=\ln Z^{\beta}_A,
\end{equation}
where $Z^{\beta}_A=\tr[\exp(-\beta\widehat{H}_A)]$. For the sandwiched R\'enyi free energy $F^\beta_{\alpha}[\n]=\beta^{-1} D_{\alpha}[\mc{N}\Vert\mc{T}^{\beta}]$, where $|A|<\infty$ and $\alpha\in[\frac{1}{2},\infty)$, we have
\begin{equation}
    \lim_{\beta\to 0^+}D_{\alpha}[\mc{N}\Vert\mc{T}^{\beta}]=D_{\alpha}[\mc{N}\Vert\mc{R}^{\pi}]=\ln|A|-S_{\alpha}[\n].
\end{equation}
If we consider the Hamiltonian of the output $A$ of the channel $\n_{A'\to A}$ to be trivial, $\widehat{H}_{A}=0$ and $\widehat{H}_{RA}=\widehat{H}_R\otimes\mathbbm{1}_A$, then for any $\beta>0$ and $\forall~\alpha\in[\frac{1}{2},\infty)$,
\begin{equation}
S_{\alpha}[\n]+\beta F^{\beta}_{\alpha}[\n]=\ln|A|.
\end{equation}
The min-entropy $S_{\infty}[\n]$ of a quantum channel $\n_{A'\to A}$ is associated with the decoupling ability of the channel~\cite[Theorem 1 \& Proposition 1]{BGC+25} and the environment-assisted erasure cost $W[A|E]_{\n}$ of the output $A$ of the channel when access to the environment $E$ is provided to the erasure~\cite[Theorem 3 \& Proposition 2]{BGC+25}. These operational interpretations of the min-entropy of a quantum channel $\n_{A'\to A}$ directly provides operational interpretations for given information processing tasks as the max-free energy of the channel $F^\beta_{\infty}[\n]=\beta^{-1}(\ln |A|-S_{\infty}[\n])$ for $\widehat{H}_A=0$. 

\section{Thermodynamical free energy and work extraction}\label{sec:energy}
For a quantum state $\rho_A$, the entropy $S(\rho)$ is $S(\rho):=S(A)_{\rho}:=-D(\rho_A\Vert\mathbbm{1}_A)$ and the energy $E(\rho)$ is $E(\rho)=\langle \widehat{H}_A\rangle_{\rho}=\tr[\widehat{H}_A\rho_A]$. For a fixed inverse temperature $\beta$, the thermodynamic non-equilibrium free energy $F^\beta_{\rm T}(\rho)$ of a quantum state $\rho_A$ is related to its entropy and energy as
\begin{equation}
    F^\beta_{\rm T}(\rho)=E(\rho)-\beta^{-1}S(\rho)=\beta^{-1}D(\rho\Vert\widehat{\gamma}^{\beta}),
\end{equation}
where $\widehat{\gamma}^{\beta}_A=\exp(-\beta \widehat{H}_A)$ is the thermal operator (unnormalized thermal state).

The entropy $S[\n]$ of a quantum channel $\n_{A'\to A}$ is defined as~\cite{DJKR06,Yua19,GW21,SPSD25} 
\begin{align}
    S[\n]&:=-D[\n\Vert\mc{R}^\mathbbm{1}]\nonumber\\
    &= \inf_{\psi\in\St(RA')}[S(RA)_{\mc{N}(\psi)}-S(R)_{\psi}],
\end{align}
where it suffices to optimize only over pure states $\psi_{RA'}$.

\begin{definition}[Thermal free energy]
    The generalized thermal free energy $\mathbf{F}^{\beta}_{\rm T}[\n]$ of an arbitrary quantum channel $\n_{A'\to A}$ is defined as
    \begin{equation}
        \mathbf{F}^\beta_{\rm T}[\n]:=\beta^{-1}\mathbf{D}[\n\|\widehat{\mc{T}}^\beta],
    \end{equation}
    where $\widehat{\mc{T}}^\beta=\mc{R}^{\widehat{\gamma}^\beta}$, i.e., $\widehat{\mc{T}}^\beta(\rho_{A'})=\tr[\rho_{A'}]\widehat{\gamma}^\beta_A$ is a replacer map that is completely positive but not trace-preserving. The thermodynamic non-equilibrium free energy $F^\beta_{\rm T}[\n]$ of a quantum channel $\n_{A'\to A}$ is $F^\beta_{\rm T}[\n]=\beta^{-1}D[\n\|\widehat{\mc{T}}^\beta]$, obtained by taking generalized channel divergence to be relative entropy.
\end{definition}

The function $\mathbf{D}[\n\|\widehat{\mc{T}}^\beta]=-\mathbf{S}^{\beta}[\n]$, where $\mathbf{S}^\beta[\n]$ is called the generalized thermal entropy of the channel $\n$~\cite{SPSD25}. We see that $\mc{F}^{\beta}_{\rm T}[\n]=\beta^{-1}\mc{S}^{\beta}[\n]$. The thermodynamic sandwiched R\'enyi free energy $F^\beta_{{\rm T},\alpha}[\n]=\beta^{-1}D_{\alpha}[\n\|\widehat{\mc{T}}^\beta]$ for $\alpha\in[\frac{1}{2},1)\cup(1,\infty]$. The sandwiched R\'enyi free energy of the absolutely thermal channel reduces to the Helmholtz free energy, 
\begin{equation}
    F^\beta_{{\rm T},\alpha}[\mc{T}^\beta]=-\beta^{-1}\ln Z^{\beta}_A,
\end{equation}
for $\alpha\in[\frac{1}{2},\infty)$, where $\alpha=1$ is taken as limit $\alpha\to 1$. For a quantum channel $\n_{A'\to A}$ with $\widehat{H}_A\propto \mathbbm{1}_A$ and noninteracting reference to the channel output, the thermodynamic sandwiched R\'enyi free energy reduces to the sandiwched R\'enyi entropy (up to a multiplicative factor of the negative inverse temperature),
\begin{equation}
    F^\beta_{{\rm T},\alpha}[\n]=-\beta^{-1}S_{\alpha}[\n].
\end{equation}

The resource-theoretic sandwiched R\'enyi free energy and the thermodynamic sandwiched R\'enyi free energy of a quantum channel $\n_{A'\to A}$ are related as
\begin{equation}
    F^\beta_{\alpha}[\n]=F^\beta_{{\rm T},\alpha}[\n]+\beta^{-1}\ln Z^\beta_A,
\end{equation}
for $\alpha\in[\frac{1}{2},\infty]$. $F^\beta_{\rm T}[\n]$ satisfies all the axiomatic properties (A1-A6) as desired from the free energy function. We also have
\begin{align}
    F^\beta[\n]& =\beta^{-1}D[\n\|\mc{T}^\beta]\nonumber\\
    &= \sup_{\psi\in\St(RA')}[F^\beta(\mc{N}(\psi_{RA'}))-F^\beta(\psi_R)],
\end{align}
where it suffices to optimize only over pure states $\psi_{RA'}$. The thermodynamic sandwiched R\'enyi free divergence $F_{{\rm T},\alpha}[\n]$, for $\alpha\in[\frac{1}{2},1)\cup(1,\infty)$ satisfies the axiomatic properties $(A1-A5)$ and quasi-convexity.

The following lemma follows from the definitions of the resource-theoretic and thermal free energies.
\begin{lemma}
    The resource-theoretic sandwiched R\'enyi free energy $F_{\alpha}[\n]$, for $\alpha\in[\frac{1}{2},\infty)$, is equal to the difference of the thermodynamic sandwiched R\'enyi free energy of the channel and the absolutely thermal channel,
    \begin{equation}
     F^\beta_{{\rm T},\alpha}[\n]-F^\beta_{{\rm T},\alpha}[\mc{T}^\beta]=F^{\beta}_{\alpha}[\n].
    \end{equation}
\end{lemma}

Based on the ``channelized" entropy and free energy definitions, we define an energy function called the thermodynamic energy $E[\n]$ of a quantum channel $\n$.
\begin{definition}
    The thermodynamic energy $E[\n]$ of a quantum channel $\n$ is defined as 
    \begin{equation}
    E[\n]:=\sup_{\psi\in\St(RA')}[E(\mc{N}(\psi_{RA'}))-E(\psi_R)],
\end{equation}
where it suffices to optimize only over pure states $\psi_{RA'}$ and $|R|=|A'|$.
\end{definition}

\begin{proposition}
    The thermodynamic energy $E[\n]$ of a quantum channel $\n_{A'\to A}$, with the reference-output Hamiltonian as $\widehat{H}_{RA}=\mathbbm{1}_R\otimes\widehat{H}_A+\widehat{H}_R\otimes\mathbbm{1}_A+\widehat{H}^{\rm int}_{RA}$, is
    \begin{equation}    E[\mc{N}]=\sup_{\psi\in\St(RA')}\left[\langle\widehat{H}_A\rangle_{\n(\psi_{A'})}+\langle\widehat{H}^{\rm int}_{RA}\rangle_{\n(\psi_{RA'})}\right],
    \end{equation}
    where it suffices to optimize over pure states $\psi_{RA'}$.
\end{proposition}
\begin{proof}
    The proof follows from the definition of $E[\n]$ and the fact that, for state $\n_{A'\to A}(\psi_{RA'})$ and Hamiltonian $\widehat{H}_{RA}$, we have
    \begin{align}
    E(\n&(\psi_{RA'}))-E(\psi_R)\nonumber\\
       = &\tr[\widehat{H}_{RA}\n(\psi_{RA'})]-\tr[\widehat{H}_R\psi_R]\nonumber\\
        =& \tr[(\widehat{H}_R\otimes\mathbbm{1}_A)\n(\psi_{RA'})]+\tr[(\mathbbm{1}_R\otimes\widehat{H}_A)\n(\psi_{RA'})]\nonumber\\
        &+\tr[\widehat{H}_{RA}^{\rm int}\n(\psi_{RA'})]-\tr[\widehat{H}_R\psi_R]\nonumber\\
        =& \tr[\widehat{H}_R\psi_R]+\tr[\widehat{H}_A\n(\psi_{A'})]+\tr[\widehat{H}^{\rm int}_{RA}\n(\psi_{RA'})]\nonumber\\
        &-\tr[\widehat{H}_R\psi_R]\nonumber\\
        =&\tr[\widehat{H}_A\n(\psi_{A'})]+\tr[\widehat{H}^{\rm int}_{RA}\n(\psi_{RA'})].
    \end{align}
\end{proof}

The following corollary directly follows from the above proposition.
\begin{corollary}\label{cor:energy}
    For an arbitrary quantum channel $\n_{A'\to A}$, when the reference-output have noninteracting Hamiltonian $\widehat{H}_{RA}=\mathbbm{1}_R\otimes\widehat{H}_A+\widehat{H}_R\otimes\mathbbm{1}_A$ ($\widehat{H}^{\rm int}_{RA}=0$), then
    \begin{equation}
        E[\n]=\sup_{\psi\in\St(A')}\langle\widehat{H}_A\rangle_{\n(\psi)},
    \end{equation}
    where it suffices to consider optimization over all pure states $\psi_{A'}$.
\end{corollary}
q

\begin{theorem}\label{thm:fets}
    The thermal free energy $F^{\beta}_{\rm T}[\n]$ of a quantum channel $\n_{A'\to A}$ is upper bounded as, when the reference $R$ is noninteracting with $A$ ($\widehat{H}_{RA}=\widehat{H}_R\otimes\mathbbm{1}_A+\mathbbm{1}_R\otimes\widehat{H}_A$),
    \begin{equation}
        F^{\beta}_{\rm T}[\n]\leq E[\n]-\beta^{-1}S[\n].
    \end{equation}
    The bound is saturated for the replacer channels $\mc{R}^\omega_{A'\to A}$, for $\omega\in\St(A)$ and $\widehat{H}^{\rm int}_{RA}=0$, $F^{\beta}_{\rm T}[\mc{R}^\omega]= E[\mc{R}^\omega]-\beta^{-1}S[\mc{R}^\omega]$. 
\end{theorem}
\begin{proof}
    It follows from the definition of the free energy $F^{\beta}[\n]$ that
    \begin{align}
      &  F^{\beta}[\n] \nonumber\\ 
        & =\sup_{\psi\in\St(RA')}\left[E(\n(\psi_{RA'}))-E(\psi_R)-\beta^{-1}(S(\n(\psi_{RA'}))\right.\nonumber\\
        & \qquad\qquad\qquad \left.-S(\psi_R))+\beta^{-1}(\ln Z^{\beta}_{RA}-\ln Z^\beta_R)\right],\nonumber
    \end{align}
    and $F^{\beta}[\n]=F^{\beta}_{\rm T}[\n]+\beta^{-1}\ln Z^\beta_A$.
    
    If $\widehat{H}_{RA}=\widehat{H}_R\otimes\mathbbm{1}_A+\mathbbm{1}_R\otimes\widehat{H}_A$, then $\ln Z^\beta_{RA}=\ln Z^\beta_R+\ln Z^\beta_A$, and
    \begin{align}
       F^{\beta}_{\rm T}[\n]& =\sup_{\psi\in\St(RA')}\left[E(\n(\psi_{RA'}))-E(\psi_R)\right.\nonumber\\
        & \quad\qquad\qquad \left.-\beta^{-1}(S(\n(\psi_{RA'}))-S(\psi_R))\right]\nonumber\\
        & \leq E[\n]-\beta^{-1}S[\n].
    \end{align}
    We conclude the proof by observing that for a replacer channel $\mc{R}^\omega$, $F^{\beta}_{\rm T}[\mc{R}^\omega]=F^{\beta}_{\rm T}(\omega)= E(\omega)-\beta^{-1}S(\omega)=E[\mc{R}^\omega]-\beta^{-1}S[\mc{R}^\omega]$, using the fact that $R$ and $A$ are noninteracting.
\end{proof}

\begin{lemma}\label{lem:choi}
    The generalized resource-theoretic free energy $\mathbf{F}^{\beta}[\n]$ and the generalized thermal free energy $\mathbf{F}^{\beta}_{\rm T}[\n]$ of an arbitrary quantum channel $\n_{A'\to A}$ upper bounds the generalized resource-theoretic free energy $\mathbf{F}^{\beta}(\Phi^{\n}_{RA})$ and the generalized thermal free energy $\mathbf{F}^{\beta}_{\rm T}(\Phi^{\n}_{RA})$ of its Choi state $\Phi^{\n}_{RA}$, respectively,
    \begin{align}
        \mathbf{F}^{\beta}[\n]& \geq \mathbf{F}^{\beta}(\Phi^{\n}_{RA}),\\
         \mathbf{F}^{\beta}_{\rm T}[\n]& \geq \mathbf{F}^{\beta}_{\rm T}(\Phi^{\n}_{RA})+\beta^{-1}\ln |A'|,
    \end{align}
    when $\widehat{H}_{RA}=\mathbbm{1}_R\otimes(\widehat{H}_A+c\mathbbm{1}_A)$ for $c\in\mathbbm{R}$.
\end{lemma}
\begin{proof}
    The proof follows from the observation that for $\widehat{H}_R=\mathbbm{1}_R\otimes(\widehat{H}_A+c\mathbbm{1}_A)$, we have $\pi_R\otimes\gamma^{\beta}_{A}=\gamma^{\beta}_R\otimes\gamma^\beta_A=\gamma^{\beta}_{RA}=\frac{1}{|A'|}\mathbbm{1}_R\otimes\mathbbm{1}_A=\frac{1}{|A'|}\widehat{\gamma}_{RA}$ for $\widehat{H}_{RA}=\mathbbm{1}_R\otimes(\widehat{H}_A+c\mathbbm{1}_A)$.
\end{proof}
\subsection{Extractable work from a quantum channel}
We consider a quantum system $A$ with Hamiltonian $\widehat{H}$ in state $\rho$ in contact with a bath at the inverse temperature $\beta$. As the heat $\delta Q$ flows from the bath into the system and an external agent performs $\delta W$ amount of work on the system, the change $\delta E=\delta\tr(\rho \widehat{H})$ in the internal energy of the system is given by the first law of thermodynamics,
\begin{equation}
    \delta E=\delta W+\delta Q.
\end{equation}

The entropy $S(\rho)$ of the system in thermodynamics is defined through its change due to the heat flow into the system $\delta Q_{\rm rev}$ during a reversible process at inverse temperature $\beta$,
\begin{equation}
    \delta S:=\beta\delta Q_{\rm rev}.
\end{equation}
A reversible, isothermal process is necessarily quasistatic and the system is always in thermodynamic equilibrium with the bath. Therefore, during a reversible process, the work done on the system is given by
\begin{equation}
    \delta W=\delta E-\beta^{-1}\delta S=\delta F_{\rm T}^{\rm eq},
\end{equation}
where $F_{\rm T}^{\rm eq}(\rho_\beta):=E(\rho_\beta)-\beta^{-1}S(\rho_\beta)$ is the equilibrium free energy of the state $\rho_\beta$ that is in thermodynamical equilibrium with the bath. An example of such a state is the Gibbs state $\gamma_\beta=\frac{e^{-\beta \widehat{H}}}{Z^{\beta}}$. The equilibrium free energy is $F^{\rm eq}_{\rm T}(\gamma_{\beta})=-\beta^{-1}\ln Z^{\beta}$.

\textit{Non-equilibrium thermodynamics}.--- If the system is not in equilibrium with the bath, the change $\delta S$ in the entropy of the system has contributions from the heat flow $\delta Q$ to the environment, which is called the reversible entropy production $\delta_eS=\beta\delta Q$, and also from an irreversible entropy production $\delta_iS$ due to change in correlations with the environment~\cite{ELV10},
\begin{align}
    \delta S=\beta\delta Q+\delta_i S.
\end{align}
The work done during an irreversible process is given by
\begin{align}
    \delta W&=\delta E-\delta Q\\
    &=\delta E-\beta^{-1}\delta S+\beta^{-1}\delta_iS\\
    &=\delta F^\beta_{\rm T}+\beta^{-1}\delta_i S,
\end{align}
where $F^\beta_{\rm T}(\rho)=E(\rho)-\beta^{-1}S(\rho)$ is the non-equilibrium free energy, also called the thermal free energy. Let us now consider the thermalization of a system with Hamiltonian $\widehat{H}$ in the non-equilibrium state $\rho$ to the Gibbs state $\gamma_\beta$. This transformation can be carried in two steps:
\begin{enumerate}
    \item \textit{Quenching}: We instantaneously change the Hamiltonian from $\widehat{H}$ to $-\beta^{-1}\ln \rho$. With respect to this new Hamiltonian, the state $\rho$ is equilibrium state. Since this is an adiabatic process, the work done on the system during the quench is the change in the energy of the system, given by
\begin{align}
    \delta W_{\rm quench}&= \delta \tr(\rho \widehat{H})\\
    &= \tr(\rho(-\beta^{-1}\ln\rho-\widehat{H}))\\
    &=-\beta^{-1}\tr(\rho\ln\rho-\rho\ln\gamma_\beta)+\beta^{-1}\ln Z^{\beta}\\
    &=-F^\beta(\rho)-F_{\rm T}^{\rm eq}(\gamma_\beta)=-F_{\rm T}^\beta(\rho).
\end{align}
where $F^\beta(\rho)=\beta^{-1}\tr(\rho\ln\rho-\rho\ln\gamma_\beta)=\beta^{-1}D(\rho\Vert\gamma_\beta)$ is the resource-theoretic free energy of the state $\rho$. 
\item \textit{Reversible, isothermal process}: We quasistatically change the Hamiltonian from $-\beta^{-1}\ln \rho$ back to $\widehat{H}$. Since the state at the initial point of this process is at equilibrium, the state at the final point of the reversible process will be $\gamma_\beta$. The work done during this reversible process is given by
\begin{align}
    \delta W_{\rm rev}&=\delta F^{\rm eq}_{\rm T}=F^{\rm eq}_{\rm T}(\gamma_\beta)-F^{\rm eq}_{\rm T}(\rho)=F^{\rm eq}_{\rm T}(\gamma_\beta).
\end{align}
\end{enumerate}
The total work done on the system during this process is
\begin{align}
    W_{\rho\to\gamma}&=\delta W_{\rm quench}+\delta W_{\rm rev}\\
    &=-F^{\beta}_{\rm T}(\rho)+F_{\rm T}^{\rm eq}(\gamma_\beta)\\
    &=-F^\beta(\rho).
\end{align}
This is also the minimum work that must be done on the system during the transformation $\rho\to \gamma^\beta$~\cite{EV11}. Equivalently, the maximum work that can be extracted from an out-of-equilibrium state $\rho$ as it thermalizes to the Gibbs state $\gamma_\beta$ is given by
\begin{align}
    W^{\rm ext}_{\rho\to \gamma_\beta}=F^\beta(\rho).
\end{align}

\textit{Partial thermalization}: We now consider a bipartite state $\psi_{RA'}$ and a quantum channel $\n_{A'\to A}$ that acts locally on $A'$, say $\rho_{RA}:=\n(\psi_{RA'})$. We assume that the channel output and its reference are noninteracting, $\widehat{H}_{RA}=\widehat{H}_R\otimes\mathbbm{1}_A+\mathbbm{1}_R\otimes\widehat{H}_A=\widehat{H}_R+\widehat{H}_A$. 
The maximum amount of work that we can extract work from $\rho_{RA}$ via partial thermalization to the state $\rho_R\otimes\gamma^\beta_A=\mathcal{T}^\beta(\psi_{RA'})$ is given by the difference in their free energies~\cite{HO13,BHN+15},
\begin{align}
W^{\rm ext}_{\n(\psi_{RA'})\to \mathcal{T}^\beta(\psi_{RA'})}&=F^\beta(\n(\psi_{RA'}))-F^\beta(\mathcal{T}^\beta(\psi_{RA'}))\\
   &=\beta^{-1}I(A;R)_{\n(\psi_{RA'})}+F^\beta(\n(\psi_{A'})).
\end{align}
We define the extractable work $W^\beta_{\rm ext}[\n]$ from a quantum channel $\n_{A'\to A}$ via partial thermalization process with respect to the bath $\beta$ as detailed above, as
\begin{equation}
    W^\beta_{\rm ext}[\n]:=\sup_{\psi\in\St(RA')}W^{\rm ext}_{\n(\psi_{RA'})\to\mc{T}^\beta(\psi_{RA'})}.
\end{equation}

In the following theorem, we establish a direct relation between the resource-theoretic free energy $F^\beta[\n]$ of a quantum channel $\n_{A'\to A}$ and its maximal extractable work via partial thermalization. 
\begin{theorem}\label{thm:excess}
The resource-theoretic free energy $F^\beta[\n]$ of a quantum channel $\n$ measures the maximal extractable work of the channel via partial thermalization,
   \begin{align}
F^\beta[\n]=W^\beta_{\rm ext}[\n].
\end{align}
\end{theorem}
The proof follows directly from Proposition~\ref{prop:mi}. The theorem quantitatively establishes the direct relation of the resource-theoretic free energy of a channel with its thermodynamic utility in terms of extractable work~\cite{BGD25b}.
\section{Discussion}\label{sec:discussion}
 We develop an asymptotically reversible resource theory of athermality for quantum channels, using Gibbs-preserving superchannels as free operations. Table~\ref{tab:concepts} summarizes a comparison of basic thermodynamic concepts for states (as known in the literature) and channels (all except entropy are introduced in this work). The (resource-theoretic and thermodynamic) free energy of a quantum channel quantitatively and operationally capture its distinguishability from the absolutely thermal channel with respect to a bath, thereby providing general and concrete form to the suggestions put forward in \cite{NG15,SPSD25,DS25}. The discussions in \cite[Section IV. A.]{DGP24} and \cite[Sections V \& VI]{SPSD25} along with the findings in this work, provide a generic framework to extend the notion of the free energy of a quantum channel to the notion of the free energy of higher order quantum processes and the conditional free energy of a bipartite quantum processes. We will be exploring these directions rigorously and formally in future works.

Another future direction is to study a resource theory of athermality of quantum channels under other classes of physically and operationally motivated free operations, for example, superchannel generalization of thermal operations~\cite{BHN+15,Gou24}. The set of all superchannels that are  superchannel generalization of thermal operations will be a subset of Gibbs preserving superchannels. We want to see if the associated resource theory will be reversible or irreversible when the uses of catalysis are allowed (cf.~\cite{ST25}). Another future direction would be to inspect the behavior of free energy functions for quantum Markovian and non-Markovian dynamics and their implications on the capabilities of the dynamics to perform thermodynamic tasks~\cite{Pot24,BGC+25,BCB+25,TBB25,BGG+25}.

\begin{acknowledgments} 
The authors thank Felix C. Binder, Manabendra Nath Bera, Swati Choudhary, Karol Horodecki, Yutong Luo, and Uttam Singh for discussions. D.G.S. thanks the Department of Science and Technology, Government of India, for the INSPIRE fellowship. S.D. acknowledges support from the Science and Engineering Research Board (now ANRF), Department of Science and Technology, Government of India, under Grant No. SRG/2023/000217 and the Ministry of Electronics and Information Technology (MeitY), Government of India, under Grant No. 4(3)/2024-ITEA. S.D. also acknowledges support from the National Science Centre, Poland, grant Opus 25, 2023/49/B/ST2/02468. S.D. thanks the University of Gda\'nsk, Gda\'nsk, Poland for the hospitality during his visit.
\end{acknowledgments}

\appendix

\section{Detailed Proofs}
For sanwiched R\'enyi relative entropy $D_{\alpha}$, $\alpha=1$ is used to denote (Umegaki) relative entropy as $\lim_{\alpha\to 1}D_{\alpha}(\cdot\|\cdot)=D(\cdot\|\cdot)$. Additionally, given two completely positive maps $\n$ and $\mc{M}$, we say $\n\ge \mc{M}$ if the linear map $\n-\mc{M}$  is also completely positive. We now prove the results claimed in the main content. 
\subsection{Proof of Theorem~\ref{thm:axioms-satisifaction}}\label{proof:axioms-satisfaction}
We begin the proof by proving two useful lemmas employed to prove the theorem.
\begin{lemma}\label{lemma:free-energy-cont}
    Given two quantum channels $\n_{A' \to A}$ and $\mathcal{M}_{A' \to A}$ such that $\frac{1}{2}\|\n-\mathcal{M}\|_{\diamond} \le \varepsilon$, we have
    \begin{align}
        \abs{F^\beta[\n] -  F^\beta[\mathcal{M}]} \le {\beta}^{-1} (\varepsilon K + h_2(\varepsilon)),
        \end{align}
where $\varepsilon \in [0,1]$ and $h_2(\varepsilon) =-\varepsilon\ln(\varepsilon)-(1-\varepsilon)\ln(1-\varepsilon)$ is the binary Shannon entropy. $K$ is a positive real number such that
        \begin{equation}
       K\ge \max \{ D_\infty(\n(\psi^0_{RA'})\|\mathcal{T}^\beta(\psi^0_{RA'})),D_\infty(\mc{M}(\phi^0_{RA'})\Vert\mc{T}^\beta(\phi^0_{RA'}))\}.
        \end{equation}
where $\psi^0_{RA'}$ maximizes $D(\n(\psi_{RA'})\|\mathcal{T}^\beta(\psi_{RA'}))$ and $\phi^0_{RA'}$ maximizes $D(\mc{M}(\psi_{RA'})\|\mathcal{T}^\beta(\psi_{RA'}))$ respectively.
\end{lemma}
\begin{proof}
The condition $\frac{1}{2}\|\n-\mathcal{M}\|_{\diamond} \le \varepsilon$ implies that $\frac{1}{2}\|\n(\psi_{RA'})-\mathcal{M}(\psi_{RA'})\|_1 \le \varepsilon$ for all $\psi_{RA'}$. The continuity of relative entropy in the first argument has been tightened in \cite{BLT25,ABD+25}. Here we use a slightly simpler bound given in \cite[Theorem 1]{BLT25}, from which it follows that
    \begin{align}
        &D(\n(\psi_{RA'})\|\mathcal{T}^\beta(\psi_{RA'})) - D(\mathcal{M}(\psi_{RA'})\|\mathcal{T}^\beta(\psi_{RA'})) \nonumber\\
        &\le \varepsilon K + h_2(\varepsilon),
    \end{align}
where 
\begin{align}
    K\ge D_\infty(\n(\psi_{RA'}) \| \mathcal{T}^\beta(\psi_{RA'}))
\end{align}
is some positive real number. Therefore, we have
\begin{align}\label{eq:n-m_continuity}
    &D[\n\|\mathcal{T}^\beta] - D[\mathcal{M}\|\mathcal{T}^\beta] \nonumber\\
    &\le D(\n(\psi^0_{RA'})\|\mathcal{T}^\beta(\psi^0_{RA'})) - D(\mathcal{M}(\psi^0_{RA'})\|\mathcal{T}^\beta(\psi^0_{RA'}))\\
    & \le \varepsilon K_{\n} + h_2(\varepsilon), 
\end{align}
where the state $\psi^0_{RA'}$ maximizes $D(\n(\psi_{RA'})\|\mathcal{T}^\beta(\psi_{RA'}))$ and $K_{\n}\ge D_\infty(\n(\psi^0_{RA'})\|\mathcal{T}^\beta(\psi^0_{RA'}))$. Following the similar steps for the channel $\mathcal{M}$,
\begin{align}\label{eq:m-n_continuity}
    D[\mathcal{M}\|\mathcal{T}^\beta] - D[\n\|\mathcal{T}^\beta]
    \le  \varepsilon  K_{\mc{M}} + h_2(\varepsilon),
\end{align}
where $K_{\mc{M}}\ge D_\infty(\mc{M}(\phi^0_{RA'})\Vert\mc{T}^\beta(\phi^0_{RA'}))$ where the state $\phi^0_{RA'}$ maximizes $D(\mc{M}(\psi_{RA'})\|\mathcal{T}^\beta(\psi_{RA'}))$. From Eqs.~\eqref{eq:n-m_continuity} and~\eqref{eq:m-n_continuity}, we have
\begin{align}
    &|D[\n\|\mathcal{T}^\beta] - D[\mathcal{M}\|\mathcal{T}^\beta]| \le  \varepsilon K + h_2(\varepsilon)\\
   &\text{where,}~ K=\max\{K_{\n},K_{\mc{M}}\}.\nonumber
\end{align}
The proof of the lemma follows by using the definition of free energy $F[\n]=\beta^{-1} D[\n\Vert\mc{T}^\beta]$ in the above inequality.
\end{proof}
\begin{lemma}\label{lemma:max-free-energy-cont}
Given two quantum channels $\n_{A' \to A}$ and $\mathcal{M}_{A' \to A}$ such that $\frac{1}{2}\|\n-\mathcal{M}\|_{\diamond} \le \varepsilon$, we have
    \begin{equation}
        \abs{F^\beta_\infty[\n] -  F^\beta_\infty[\mathcal{M}]}
        \le \beta^{-1} \ln (1+\varepsilon |A'|Z^\beta \mathrm{e}^{\beta E_{\max}}),
    \end{equation}
where $E_{\max}$ is the maximum eigenvalue of the Hamiltonian $\widehat{H}_A$.  
\end{lemma}
\begin{proof}
Using the properties of max-relative channel entropy, we have,
\begin{align}
    D_\infty&[\n\Vert\mc{T}^\beta]-D_\infty[\mc{M}\Vert\mc{T}^\beta]\\
    &=D_\infty(\Phi^\n\Vert\Phi^{\mc{T}^\beta})-D_\infty(\Phi^{\mc{M}}\Vert\Phi^{\mc{T}^\beta})\\
    &=D_\infty(\Phi^\n\Vert\pi_R\otimes\gamma^\beta_A)-D_\infty(\Phi^{\mc{M}}\Vert\pi_R\otimes\gamma^\beta_A).
\end{align}
The minimum eigenvalue of the second argument is given by
\begin{align}
    \lambda_0:=\lambda_{\min}(\pi_R\otimes\gamma^\beta_A)=\frac{\mathrm{e}^{-\beta E_{\max}}}{|A'|Z^{\beta}},
\end{align}
where $E_{\max}$ is the largest eigenvalue of the Hamiltonian $\widehat{H}_A$. Using the bounds for uniform continuity of max-relative entropy given in~\cite{BCGM24}, we have
\begin{align}
    | D_\infty&[\n\Vert\mc{T}^\beta]-D_\infty[\mc{M}\Vert\mc{T}^\beta]|\\
    &\le \min 
        \begin{cases}
        \ln (1+\varepsilon)+ \ln \lambda_0^{-1}, \\
        \ln (1+\varepsilon \lambda_0^{-1}), \\
        \ln ((1+\varepsilon)(1+\varepsilon \lambda_0^{-1})- \varepsilon^2).
        \end{cases} 
\end{align}
Note that $Z^{\beta} \mathrm{e}^{\beta E_{\max}}> 1$ for any Hamiltonian $\widehat{H}_A$. Therefore, $\lambda_0^{-1}$ is greater than $1$, which makes the second term in the above cases the minimum value for any Hamiltonian. Therefore, we have 
\begin{align}
    | D_\infty&[\n\Vert\mc{T}^\beta]-D_\infty[\mc{M}\Vert\mc{T}^\beta]|\le   \ln (1+\varepsilon |A'|Z^{\beta}\mathrm{e}^{\beta E_{\max}}).
\end{align}
The proof follows from the definition of the max-free energy.
\end{proof}

We now prove the theorem.
\begin{theorem*}
    Let $\mc{N}_{A'\to A}$ be an arbitrary quantum channel. The free energy $F^{\beta}[\mc{N}]$ satisfies all the thermodynamical (resource-theoretic) axiomatic properties (A1-A6).
    
    For $\alpha\in(1,\infty)$, the sandwiched R\'enyi free energy $F_{\alpha}^{\beta}[\mc{N}]$ satisfies (A1-A5) and quasi-convexity, i.e., for quantum channel $\mc{N}=\sum_x p_x\mc{N}^x$ where $\{p_x\}_x$ is some probability distribution, $F_{\alpha}^{\beta}[\mc{N}]\leq \max_xF_{\alpha}^{\beta}[\mc{N}^x]$. 
\end{theorem*}
\begin{proof}
The free energy of a channel satisfies all the axioms as shown below.\\
\noindent
(A1) \textit{Monotonicity under Gibbs preserving superchannels}: This follows from the monotonicity of sandwiched R\'enyi relative entropy between channels for $\alpha\in[\frac{1}{2},\infty)$~\cite{G19},
\begin{align}
D_{\alpha}[\n_{A' \to A} \| \mathcal{T}^{\beta}_{A' \to A}] & \geq D_{\alpha}[\Theta (\n_{A' \to A}) \| \Theta( \mathcal{T}^{\beta}_{A' \to A})]\\
    &= D_{\alpha}[\Theta (\n_{A' \to A}) \|  \mathcal{T}^{\beta}_{A' \to A}],
\end{align}
which imply $ F_{\alpha}^\beta[\Theta (\n_{A' \to A})] \le F_{\alpha}^\beta[\n_{A' \to A}]$.
\\
\noindent
(A2) \textit{Reduction to states}:  For a replacer channel $\mathcal{R}^{\omega}_{A'\to A}$, for all $\alpha \in [\frac{1}{2},\infty)$, we have
        \begin{align}
            \beta F^\beta_\alpha[\mathcal{R}^{\omega}]&=D_\alpha[\mathcal{R}^\omega\Vert\mathcal{T}^{\beta}]\nonumber\\
            &= \sup_{\psi\in\St(RA')}D_\alpha(\mathcal{R}^{\omega}(\psi_{RA'})\Vert\mathcal{T}^\beta(\psi_{RA'}))\nonumber\\
            &=\sup_{\psi\in\St(RA')}D_\alpha(\psi_R\otimes\omega_A\Vert\psi_R\otimes\gamma^\beta_A)\nonumber\\
            &= D_\alpha(\omega_A\Vert\gamma^\beta_A)=\beta F^\beta_\alpha(\omega).
        \end{align}
(A3) \textit{Minimum value at absolutely thermal channel}: This follows from the fact that  for $\alpha \in [\frac{1}{2},\infty)$, $D_\alpha(\rho\Vert\sigma)\ge 0$ if both $\rho$ and $\sigma$ are states, and $D_\alpha(\rho\|\sigma)= 0$ if and only if $\rho=\sigma$. Therefore, given quantum channels $\n_{A'\to A}$ and $\mc{M}_{A'\to A}$, $D_\alpha[\mc{N}\Vert\mc{M}]=\sup_{\psi\in\St(RA')} D_\alpha(\n(\psi)\|\mathcal{M}(\psi))=0$ if and only if $\n(\varphi)=\mathcal{M}(\varphi)$ for all $\varphi\in \St(RA')$ or $\n=\mathcal{M}$. Therefore, $F^\beta_\alpha[\n]=\beta ^{-1} D_\alpha[\n\|\mathcal{T}^\beta]\ge 0$ with equality if and only if $\n=\mathcal{T}^\beta$.
\\
\noindent
(A4) \textit{Continuity:} We show the continuity of channel free energy $F^\beta[\n]$ in Lemma~\ref{lemma:free-energy-cont} and channel max-free energy $F^\beta_\infty[\n]$ in Lemma~\ref{lemma:max-free-energy-cont}. The continuity of $F_\alpha^\beta[\n]$ for $\alpha\in[\frac{1}{2},\infty)$ can be shown similarly using the continuity of $D_\alpha(\rho\Vert\sigma)$ in the first argument, as given in~\cite{BCGM24}.
\\
\noindent
(A5) \textit{Additivity}: For two channels $\n_{A_1' \to A_1}$ and $\mc{M}_{A_2' \to A_2}$ such that the output interaction Hamiltonian $\widehat{H}_{A_1A_2}$ is zero, the absolutely thermal channel can be written as $\mc{T}^\beta_{A_1'A_2'\to A_1A_2}=\mc{T}^\beta_{A_1' \to A_1} \otimes \mc{T}^\beta_{A_2' \to A_2}=:\mc{T}^\beta_1 \otimes \mc{T}^\beta_2$.  In this case, the additivity of the generalized free energy of channel follows from the additivity of the generalized channel divergence since
\begin{equation}
    \beta F^\beta_\alpha[\n \otimes \mc{M}] = D_\alpha[\n \otimes \mc{M}\|\mathcal{T}^\beta_1 \otimes \mathcal{T}^\beta_2].
\end{equation}
We will show that the sandwiched R\'enyi channel divergence is additive under tensor product with respect to the absolutely thermal channel, i.e. 
\begin{align}
    &D_\alpha[\n \otimes \mc{M}\|\mathcal{T}^\beta_1 \otimes \mathcal{T}^\beta_2]\\
    &= D_\alpha[\n\|\mathcal{T}^\beta_1] + D_\alpha[\mc{M}\| \mathcal{T}^\beta_2].
\end{align}
The proof of the inequality "$\ge$" follows from the definition of the sandwiched R\'enyi channel divergence,
\begin{align}
    &D_\alpha[\n \otimes \mc{M}\|\mathcal{T}^\beta_1 \otimes \mathcal{T}^\beta_2] \nonumber \\
    &= \sup_{\psi\in\St(RA_1'A_2')} D_\alpha(\n \otimes \mc{M} (\psi)\| \mathcal{T}^\beta_1 \otimes \mathcal{T}^\beta_2(\psi))\\
    & \ge  \sup_{\varphi\in\St(R_1A_1')\otimes\St(R_2A_2')} D_\alpha(\n \otimes \mc{M} (\varphi)\| \mathcal{T}^\beta_1 \otimes \mathcal{T}^\beta_2(\varphi)) \\
    &= D_\alpha[\n\|\mathcal{T}^\beta_1] + D_\alpha[\mc{M}\| \mathcal{T}^\beta_2],
\end{align}
where we have used the additivity of the sandwiched R\'enyi divergence of states under the tensor product for $\alpha \in [\frac{1}{2},\infty)$. 

To prove the inequality "$\le$", we use the following result from \cite[Lemma 38 and Proposition 41]{WBHK20}, for all $\alpha \in [1,\infty)$, for an arbitrary state $\rho_{RA'}$, a channel $\n_{A'\to A}$ , and a replacer channel $\mc{R}^{\tau}_{A'\to A}$ which replaces all its input states with a state $\tau$,
\begin{align}
    &D_\alpha(\n(\rho_{RA'})\|\sigma_{R} \otimes \tau_A) \\
    &\le D_\alpha[\n\|\mathcal{R}^\tau] + D_\alpha(\rho_{RA'}\|\sigma_{RA'}).
\end{align}
Let $\psi_{RA_1'A_2'}$ be an arbitrary pure state and states $ \rho_{R'A_1'}$ and $\sigma_{R'A_1'}$ be such that $R'=RA_2$, i.e., $\mc{M}(\psi_{RA_1'A_2'})=\rho_{R'A_1'}$ and $\mathcal{T}^\beta_{A_2' \to A_2}(\psi_{RA_1'A_2'})=\sigma_{R'A_1'}$,
\begin{align}
    &D_\alpha(\n \otimes \mc{M} (\psi_{RA_1'A_2'})\| \mathcal{T}^\beta_{A_1' \to A_1} \otimes \mathcal{T}^\beta_{A_2' \to A_2}(\psi_{RA_1'A_2'})) \\
    &\le D_\alpha[\n\|\mathcal{T}^\beta] + D_\alpha(\mc{M}(\psi_{RA_1'A_2'}) \| \mathcal{T}^\beta_{A_2' \to A_2}(\psi_{RA_1'A_2'}))\\
    & \le D_\alpha[\n\|\mathcal{T}^\beta] + D_\alpha[\mc{M}\|\mathcal{T}^\beta].
\end{align}
This implies the additivity of  $F_\alpha^\beta[\n]$ for $\alpha\in[1,\infty)$.
\\
\noindent
(A6) \textit{Joint quasi-convexity}: It is known that the sandwiched R\'enyi divergence is jointly convex for all $\alpha \in [\frac{1}{2},1]$ and jointly quasi-convex for all $\alpha \in (1,\infty)$. For $\alpha\in[\frac{1}{2},1]$, we have
\begin{align}
  D_\alpha&\left[\sum_x p_x \n^x\|\mathcal{T}^\beta\right] \nonumber\\  
    & = \sup_{\psi_{RA'}} D_\alpha\left(\sum_x p_x \n^x(\psi)\| \mathcal{T}^\beta(\psi)\right)\\
    & \le \sup_{\psi_{RA'}} \sum_x p_x D_\alpha( \n^x(\psi)\| \mathcal{T}^\beta(\psi)).
\end{align}
Using the definition of free energy of channels, we have 
\begin{align}
    F_\alpha^\beta\left[\sum_xp_x\n^x\right]\le\sum_xp_xF_\alpha^\beta[\n^x]~\text{for}~\alpha\in \left[\frac{1}{2},1\right].
\end{align}
Similarly, for $\alpha \in (1,\infty)$, using the quasi-convexity of sandwiched R\'enyi divergence, we have the quasi-convexity for the sandwiched R\'enyi free energy,
\begin{equation}
 F^\beta_\alpha \left[\sum_x p_x \n^x\right] \le\max_x F^\beta_\alpha [\n^x].
\end{equation}
\end{proof}
\subsection{Proof of Lemma \ref{lem:monotonicity_subpreserving} \label{app:proof_lem_monotonicity_subpreserving}}
\begin{lemma*}
    The sandwiched R\'enyi free energy $F^{\beta}_{\alpha}[\mc{N}]$, $\alpha\in[\frac{1}{2},\infty)$, is nonincreasing under the action of Gibbs-subpreserving superchannel. It remains invariant under the action of Gibbs preserving unitary superchannels.
\end{lemma*}
\begin{proof}
Given the channel $\n_{A'\to A}$ and absolutely thermal channel $\mc{T}^\beta_{A'\to A}$,  the sandwiched R\'enyi relative entropy of a channel is monotonically non-increasing under the action of a superchannel $\Theta\in\SCh((A',A),(B',B))$ \cite{G19}. i.e. for all $\alpha \in [\frac{1}{2}, \infty)$,
    \begin{equation}\label{alpha_monoton_1}
       D_{\alpha}[\Theta (\n) \| \Theta( \mathcal{T}^{\beta})] \le D_{\alpha}[\n\| \mathcal{T}^{\beta}].
    \end{equation}
For positive semidefinite operators $\rho \neq 0$, $\sigma$ and $\sigma'$, if $\sigma' \ge \sigma$, we have $D_{\alpha} (\rho \| \sigma') \le D_{\alpha} (\rho \| \sigma)$ for $\alpha \in [\frac{1}{2}, \infty)$ (\cite[Proposition 4]{MDS+13}).
Therefore, under the action of a $\mathcal{T}^{\beta}$-subpreserving superchannel, we have
\begin{align*}
    &D_{\alpha}(\Theta(\n)(\rho_{RB'}) \| \mathcal{T}^\beta_{B' \to B}(\rho_{RB'})) \nonumber \\
    & \le D_{\alpha}(\Theta(\n)(\rho_{RB'}) \| \Theta(\mathcal{T}^\beta_{A'\to A})(\rho_{RB'})),
\end{align*}
for all states $\rho_{RB'}$. Taking the supremum over all states on both sides,
\begin{equation} \label{alpha_monoton_2}
     D_{\alpha}[\Theta(\n) \| \mathcal{T}^\beta_{B'\to B}] \le D_{\alpha}[\Theta(\n) \| \Theta(\mathcal{T}^\beta_{A'\to A})] .
\end{equation}
Using Eqs.~\eqref{alpha_monoton_1} and~\eqref{alpha_monoton_2}, we have 
\begin{equation}
     D_{\alpha}[\Theta(\n) \| \mathcal{T}^\beta_{B'\to B}] \le D_{\alpha} [\n \| \mathcal{T}^\beta_{A'\to A}],
\end{equation}
for all $\alpha \in [\frac{1}{2}, \infty)$. This implies the first part of the lemma. 

The invariance of $F_\alpha^\beta[\n]$ for $\alpha\in[\frac{1}{2},\infty)$ under Gibbs preserving unitary superchannel  follows directly from \cite[Proposition 1]{DGP24}.
\end{proof}
\subsection{Proof of Proposition~\ref{thm:unitary-free-energy}}\label{app:proof-thm-unitary-free-energy}
\begin{lemma}\label{lemma:convex_purestate}
    Given a state $\rho=\sum_ip_i\ket{\psi_i}\bra{\psi_i}$ where $\psi_i$ are pure states, the max-relative entropy is bounded as follows
    \begin{align}
D_\infty(\rho\Vert\sigma)\ge \ln(\max_ip_i\bra{\psi_i}\sigma^{-1}\ket{\psi_i}),\\
D_\infty(\rho\Vert\sigma)\le\ln(\sum_ip_i\bra{\psi_i}\sigma^{-1}\ket{\psi_i}).
    \end{align}
\end{lemma}
\begin{proof}
The max-relative entropy $D_\infty(\rho\Vert\sigma)$ can be written as
    \begin{align}
        D_\infty(\rho\Vert\sigma)&=\ln\Vert\sum_ip_i\sigma^{-\frac{1}{2}}\ket{\psi_i}\bra{\psi_i}\sigma^{-\frac{1}{2}}\Vert_\infty\\
        &=\ln\Vert AA^\dagger\Vert_\infty\\
        &=\ln\lambda_{\max}(AA^\dagger),
    \end{align}
 where the matrix $A$ is given by $A=[\ket{v_1}~~\ket{v_2}~...]$ with the vectors $\ket{v_i}:=\sqrt{p_i}\sigma^{-\frac{1}{2}}\ket{\psi_i}$. The corresponding Gram matrix of matrix $A$ is given by $G=A^\dagger A$, such that $G_{ij}=\braket{v_i}{v_j}=\sqrt{p_ip_j}\bra{\psi_i}\sigma^{-1}\ket{\psi_j}$. Note that $AA^\dagger$ and $G$ have the same non-zero eigenvalues. Therefore, $D_\infty(\rho\Vert\sigma)=\ln(\lambda_{\max}(AA^\dagger))=\ln(\lambda_{\max}(G))$. Now, for any matrix $G$ we have the following property
 \begin{align}
     \max_{i}(G_{ii})\le\lambda_{\max}(G)\le \tr(G).
 \end{align}
 Using the fact that $G_{ij}=\braket{v_i}{v_j}=\sqrt{p_ip_j}\bra{\psi_i}\sigma^{-1}\ket{\psi_j}$, we have the proof of the stated lemma.
\end{proof}
\begin{proposition*}
    For $\alpha\in\{1,\infty\}$ and $|A'|=|A|$, the sandwiched R\'enyi free energy $F^{\beta}_{\alpha}[\mc{N}]$ of a quantum channel $\n_{A'\to A}$ is maximum if and only if $\n$ is a unitary channel. For all unitary channels $\mc{U}_{A'\to A}$, $F^{\beta}_{\alpha}[\mc{U}]=F^{\beta}_{\alpha}[\id_{A'\to A}]$.
\end{proposition*}
\begin{proof}
Given a state $\rho_{RA}$ and its decomposition into pure states $\rho_{RA}=\sum_xp_x\psi^x_{RA}$, we use the joint quasi-convexity of the sandwiched R\'enyi entropy to observe that for $\alpha\in[1,\infty)$
\begin{align}\label{eq:joint-quasi-convex}
D_\alpha(\sum_xp_x\psi^x_{RA}\Vert\sum_xp_x\rho^x_{R}\otimes\gamma^\beta_A)\le \max_x D_\alpha(\psi^x_{RA}\Vert\rho^x_{R}\otimes\gamma^\beta_A),
\end{align}
where $\rho^x_{R}=\tr_A\psi_{RA}^x$.  Therefore, $D_{\alpha}(\rho_{RA}\Vert\rho_R\otimes\gamma_A)$ achieves the maximum value, for $\alpha\in[1,\infty)$, if $\rho_{RA}$ is a pure state. This implies that $D_\alpha[\n\Vert\mc{T}^\beta]$ achieves the maximum value if $\n$ is a unitary channel.

We will show that the "only if" condition is satisfied for $\alpha= 1$ and $\infty$ by proving that the equality in Eq.~\eqref{eq:joint-quasi-convex} holds if and only if $\rho_{RA}$ is a pure state.

For $\alpha= 1$, relative entropy is jointly convex, and therefore for $\rho=\sum_i p_i\psi_i$, where $\psi_i$ are pure, we have   
\begin{align}
    D(\rho\Vert\sigma)\le\sum_ip_iD(\psi_i\Vert\sigma).
\end{align}
The equality in the above holds if and only if $\rho$ is pure~\cite{Rus02}.

For the case of $\alpha\to \infty$, let the equality for the quasi-convexity condition be satisfied, i.e.  $D_\infty(\rho\Vert\sigma)=\max_iD_\infty(\psi_i\Vert\sigma)$, then using Lemma~\ref{lemma:convex_purestate},
\begin{align}
  \max_i \ln(\bra{\psi_i}\sigma^{-1}\ket{\psi_i})&\le\ln(\sum_ip_i\bra{\psi_i}\sigma^{-1}\ket{\psi_i})\\
  \implies  \max_i\bra{\psi_i}\sigma^{-1}\ket{\psi_i}&\le\sum_ip_i\bra{\psi_i}\sigma^{-1}\ket{\psi_i}.
\end{align}
Since the weighted average cannot be greater than the largest element of the set, the above condition is satisfied if and only if $\rho$ is a pure state.
  
\end{proof}
\subsection{Proof of Theorem~\ref{thm:dist_cost}}\label{app:proof_thm_dist_cost}
\begin{theorem*}
For any error $\varepsilon\in[0,1]$ and a given resource channel $(\n,\mathcal{T}^\beta)$, the single-shot athermality distillation and formation are proportional to the $\varepsilon$-hypothesis-testing free energy and the $\varepsilon$-max-free energy of the channel $\n$, respectively,
\begin{align}
    \mathrm{Dist}^\varepsilon (\n,\mathcal{T}^\beta)&=\frac{1}{2}D_{H}^{\varepsilon^2}[\n\Vert\mathcal{T}^\beta],\\
     \mathrm{Cost}^\varepsilon (\n,\mathcal{T}^\beta)&=\frac{1}{2}D_{\infty}^\varepsilon[\n\Vert\mathcal{T}^\beta],
\end{align}
$F^{\beta,\varepsilon}_H[\n]=\frac{2}{\beta}\mathrm{Dist}^{\sqrt{\varepsilon}} (\n,\mathcal{T}^\beta)$ and $ F^{\beta,\varepsilon}_{\infty}[\n]=\frac{2}{\beta}\mathrm{Cost}^\varepsilon (\n,\mathcal{T}^\beta)$.
\end{theorem*}
\begin{proof}
A Gibbs preserving superchannel (GPSC) $\Theta_{\psi}^\Lambda$, defined through a pure state $\psi$ of composite system (reference and channel input) and an effect operator (element of a POVM) $0\le \Lambda\le \mathbbm{1}$, that acts on the resource $(\n,\mathcal{T}^\beta)$ to yield $(\id_m,\mathcal{R}^\pi)$ can be written as the following preparation,
\begin{align}
    \Theta^\Lambda_{\psi}(\n)=\tr(\n(\psi)\Lambda)\id+\tr(\n(\psi)(\mathbbm{1}_A-\Lambda))\frac{1}{m^2-1}\id^{\perp}.
\end{align}
The Gibbs preserving constraint on the superchannel, i.e. $\Theta_\psi^\Lambda(\mathcal{T}^\beta)=\mathcal{R}^\pi$, can be written as $\tr(\mathcal{T}^\beta(\psi)\Lambda)=\frac{1}{m^2}$, and the conversion distance $d_{\Theta_\psi^\Lambda}((\n,\mathcal{T}^\beta)\rightarrow(\id,\mathcal{R}^\pi))$ is given by $\sqrt{1-\tr(\n(\psi)\Lambda)}$. Recalling the definition of the $\varepsilon$-single-shot distillation, we have the following:
\begin{align}
        &\mathrm{Dist}^\varepsilon (\n,\mathcal{T}^\beta)\nonumber \\
        &=\sup_{m,\Lambda,\psi}\{\ln m:~  \sqrt{1-\tr(\n(\psi)\Lambda)}\le\varepsilon,~\tr(\mathcal{T}^\beta(\psi)\Lambda)=\frac{1}{m^2}\}\\
        &=\sup_{\psi,\Lambda}\left\{-\frac{1}{2}\ln\tr(\mathcal{T}^\beta(\psi)\Lambda) :~  1-\tr(\n(\psi)\Lambda)\le\varepsilon^2\right\}\\
        &=\sup_{\psi}\left[-\inf_{\Lambda}\left\{\frac{1}{2}\ln\tr(\mathcal{T}^\beta(\psi)\Lambda) :~  1-\tr(\n(\psi)\Lambda)\le\varepsilon^2\right\}\right]\\
&=\sup_{\psi}\frac{1}{2}D_H^{\varepsilon^2}(\n(\psi)\Vert\mathcal{T}^\beta(\psi))\\
         &=\frac{1}{2}D_H^{\varepsilon^2}[\n\Vert \mathcal{T}^\beta].
    \end{align}

Now, to derive the $\varepsilon$-single-shot cost, we first calculate the conversion distance: 
 \begin{align}
      &  d_{\mathrm{GP}}((\id_m\mathcal{R}^\pi) \rightarrow(\n,\mathcal{T}^{\beta}))\nonumber\\
        &\qquad=\min_{\Theta}\left\{P[\mc{N},\Theta(\id_m)]~:~\Theta(\mathcal{R}^{\pi})=\mathcal{T}^{\beta}\right\}\\
         &\qquad=\min_{\Theta}\left\{P[\mc{N},\Theta(\id_m)]~:\right.\nonumber\\ 
 &\hspace{2.5cm} \left.\mathcal{T}^{\beta}= \frac{1}{m^2}\Theta(\id_m)+\frac{1}{m^2}\Theta(\id_m^\perp)\right\}\\
         &\qquad=\min_{\mc{E}}\left\{P[\mc{N},\mc{E}]~:~m^2\mathcal{T}^{\beta}\ge \mc{E}\right\}\\
         &\qquad=\min_{\mc{E}}\left\{P[\mc{N},\mc{E}]~:~2\ln m\ge D_{\infty} [\mc{E}\Vert\mc{T}^\beta]\right\}.
    \end{align}
where we have replaced $\mc{E}=\Theta(\id_m)$. Note that $\mc{E}$ can be any channel such that $m^2\mc{T}^\beta\ge \mc{E}$, and is not constrained by the Gibbs preserving property of $\Theta$. This is because the quantity that we are optimizing is a function of the variable $\Theta(\id_m)$, and the range of this variable is fully captured by the constraint $m^2\mc{T}^\beta\ge \Theta(\id_m)$. The $\varepsilon$-single-shot cost can be written as
   \begin{align}
        &\mathrm{Cost}^\varepsilon (\n,\mathcal{T}^\beta)\nonumber \\
        &=\inf_{m}\left\{\ln m:~d_{\mathrm{GP}}((\id_m\mathcal{R}^\pi) \rightarrow(\n,\mathcal{T}^{\beta})) \le\varepsilon\right\}\\
        &=\inf_{m}\{\ln m:~ P[\mc{N},\mc{E}]\le\varepsilon,~2\ln m\ge D_{\infty} [\mc{E}\Vert\mc{T}^\beta]\}\\
        &=\inf_{\mc{E}}\left\{\frac{1}{2}D_{\max}[\mathcal{E}\Vert\mathcal{T}^\beta]:~P[\mc{N},\mc{E}] \le\varepsilon\right\}\\
         &=\frac{1}{2}D^{\varepsilon}_{\max}[\n\Vert\mathcal{T}^\beta].
    \end{align}
\end{proof}
\subsection{SDP for the max-free energy}\label{app:sdp}
The max-free energy $F_{\infty}^\beta[\n]$ of a quantum channel $\n_{A'\to A}$ is half of the logarithm of the optimal value of the following semidefinite program (SDP) (strong duality):
\begin{align*}
& \underline{\textit{Primal}}: \\[6pt]
& \begin{cases}
\text{minimize} & \lambda\in\mathbbm{R} \\[6pt]
\text{subject to} & \Phi^{\mc{N}}_{RA} \le \lambda(\pi_R \otimes \gamma_A^\beta), \\[6pt]
& \lambda \ge 0.
\end{cases}
\\[12pt]
& \underline{\textit{Dual}}: \\[6pt]
& \begin{cases}
\text{maximize} & \operatorname{tr}(\Phi^\n_{RA} X_{RA}) \\[6pt]
\text{subject to} & \operatorname{tr}((\pi_A \otimes \gamma_A^\beta) X_{RA}) \le 1, \\[6pt]
& X_{RA} \ge 0.
\end{cases}
\end{align*}
It follows from the fact that the max-relative entropy between channels is an SDP where strong duality holds~\cite{LKDW18}, i.e., optimal value of the primal and dual problems are the same.
\subsection{Proof of Lemma~\ref{prop:unitary_free_energy} \label{proof:prop:unitary_free_energy}}
\begin{lemma*}
The max-free energy $F^{\beta}_{\infty}[\mc{U}]$ of a unitary quantum channel $\mc{U}_{A'\to A}$ is
    \begin{equation}
        F^{\beta}_{\infty}[\mc{U}]= \beta^{-1}\ln \tr\left[{(\gamma^{\beta}_A)}^{-1}\right].
    \end{equation}
    If the Hamiltonian of $A$ is trivial or in general $\widehat{H}_A=c\mathbbm{1}_A$ for $c\in\mathbbm{R}$, then $F^{\beta}_{\infty}[\mc{U}]=2\beta^{-1}\ln|A|$.
\end{lemma*}
\begin{proof}
Using the property of the max-relative entropy of quantum channels
\begin{equation}\label{eq:max-relequivalence}
\beta F_{\infty}^\beta [\mc{U}] = D_{\infty}[\mathcal{U} \| \mathcal{T}^{\beta}] = D_{\infty}(\Phi^{\mathcal{U}} \|\Phi^{\mathcal{T}^\beta}),
\end{equation}
where, $\Phi^\mc{U}$ is a pure state and $\Phi^{\mc{T}^\beta}=\pi_R\otimes\gamma_{\beta}$. Given a pure state $\psi$, we have $D_\infty(\psi\Vert\sigma)=\ln\bra{\psi}\sigma^{-1}\ket{\psi}$ for a positive operator $\sigma$. Let $ \mathcal{U}(\rho) = U \rho U^\dagger$, then
\begin{align}
\beta F_{\infty}^\beta [\mc{U}]  &= D_{\infty}(\Phi^{\mathcal{U}} \|\Phi^\mathcal{T}) \\
 &= \ln \left(\frac{1}{d} \sum_{j} 
(\bra{j} \otimes \bra{j} U^\dagger) 
\left( \frac{1}{d} \mathbbm{1}_R\otimes \gamma_\beta \right)^{-1} \right. \nonumber \\
& \qquad \qquad \left. \sum_{i} (\ket{i} \otimes U \ket{i}) \right) \nonumber \\
&= \ln \left( \sum_{i,j} 
\delta_{ji}  \bra{j} U^\dagger \gamma_\beta^{-1} U \ket{i} \right) \nonumber \\
&= \ln \left( \tr \left[ \gamma_\beta^{-1} \right] \right).
\end{align}
If $\widehat{H}_A=c \mathbbm{1}_A$, $ \tr \left[ (\gamma^\beta_A)^{-1} \right] = \tr (|A| \mathbbm{1}_A) = |A|^2$ and $F^{\beta}_{\infty}[\mc{U}]=2\beta^{-1}\ln|A|$.
\end{proof}

\subsection{Free energy of the golden unit}\label{app:lemma-free-max-energy}
\begin{lemma}\label{lemma:free-max-energy}
Given a unitary channel $\mc{U}_{A'\to A}$, the minimum value of its max-free energy $F^\beta_\infty[\mc{U}]$ is achieved for the output with trivial Hamiltonian, $\widehat{H}_A=c\mathbbm{1}_A$. In particular,
    \begin{align}
\min_{\widehat{H}_A}F^\beta_\infty[\mc{U}]=2\beta^{-1}\ln|A|.
    \end{align}
\end{lemma}
\begin{proof}
From Lemma~\ref{prop:unitary_free_energy} we have 
\begin{align}
\beta F_\infty^\beta[\mc{U}]&=\ln\tr(\gamma_\beta^{-1})\\
&=\ln \tr {\operatorname{e}}^{\beta \widehat{H}_A}+\ln\tr {\operatorname{e}}^{-\beta \widehat{H}_A}.
\end{align}
We analyze the vanishing conditions of the directinanl derivative with respect to $\widehat{H}_A$, that is $\nabla_{\widehat{H}_A}F_\infty^\beta[\mc{U}]=0$, which implies that
\begin{align}
 \frac{{\operatorname{e}}^{\beta\widehat{H}_A}}{\tr({\operatorname{e}}^{\beta \widehat{H}_A})}-\frac{{\operatorname{e}}^{-\beta\widehat{H}_A}}{\tr({\operatorname{e}}^{-\beta \widehat{H}_A})}=0.
\end{align}
For all eigenvalues $E_i$ of $\widehat{H}_A$, we have
\begin{align}
    \frac{{\operatorname{e}}^{\beta E_i}}{\tr({\operatorname{e}}^{\beta \widehat{H}_A})}&=\frac{{\operatorname{e}}^{-\beta E_i}}{\tr({\operatorname{e}}^{-\beta \widehat{H}_A})}\\
    \text{or},~
{\operatorname{e}}^{2\beta E_i}&=\frac{\tr({\operatorname{e}}^{\beta \widehat{H}_A})}{\tr({\operatorname{e}}^{-\beta \widehat{H}_A})},~\forall i.
\end{align}
Since the right-hand side is independent of $i$, we conclude that all the eigenvalues of $\widehat{H}_A$ must be the same, which corresponds to the Hamiltonian of the form $\widehat{H}_A=c\mathbbm{1}_A$. For such Hamiltonian, we have $F_\infty^\beta[\mc{U}]=2\beta^{-1}\ln|A|$.

To see that this value is the global minimum for $F_{\infty}^\beta[\mc{U}]$, note that
\begin{align}
    {\rm e}^{\beta F_\infty^\beta[\mc{U}]}&=\tr {\operatorname{e}}^{\beta \widehat{H}_A}\tr\operatorname{e}^{-\beta \widehat{H}_A}\\
    &\ge (\tr {\operatorname{e}}^{\frac{1}{2}\beta \widehat{H}_A}{\operatorname{e}}^{-\frac{1}{2}\beta \widehat{H}_A})^2\\
    &=(\tr\mathbbm{1}_A)^2=|A|^2.
\end{align}
We have used the Cauchy-Schwarz inequality for Hilbert-Schmidt inner product: $(\tr(A^\dagger B))^2\le\tr(A^\dagger A)\tr(B^\dagger B)$. We conclude that $F_\infty^\beta[\mc{U}]\ge 2\beta^{-1}\ln|A|$ and the output Hamiltonian $\widehat{H}_A=c\mathbbm{1}_A$ gives the minimum free energy for unitary channels.
\end{proof}
\bibliography{output}

@article{NG15,
  title = {Nonthermal Quantum Channels as a Thermodynamical Resource},
  author = {Navascu\'es, Miguel and Garc\'{\i}a-Pintos, Luis Pedro},
  journal = {Physical Review Letters},
  volume = {115},
  issue = {1},
  pages = {010405},
  numpages = {5},
  year = {2015},
  month = {Jul},
  publisher = {American Physical Society},
  doi = {10.1103/PhysRevLett.115.010405},
  url = {https://link.aps.org/doi/10.1103/PhysRevLett.115.010405}
}

@article{GW21,
  title = {Entropy of a quantum channel},
  author = {Gour, Gilad and Wilde, Mark M.},
  journal = {Physical Review Research},
  volume = {3},
  issue = {2},
  pages = {023096},
  numpages = {26},
  year = {2021},
  month = {May},
  publisher = {American Physical Society},
  doi = {10.1103/PhysRevResearch.3.023096},
  url = {https://link.aps.org/doi/10.1103/PhysRevResearch.3.023096}
}

@misc{JGW25,
      title={Fundamental work costs of preparation and erasure in the presence of quantum side information}, 
      author={Kaiyuan Ji and Gilad Gour and Mark M. Wilde},
      year={2025},
      eprint={2503.09012},
      archivePrefix={arXiv},
      primaryClass={quant-ph},
      url={https://arxiv.org/abs/2503.09012}, 
}

@article{LR12,
  title = {One-Shot Classical-Quantum Capacity and Hypothesis Testing},
  author = {Wang, Ligong and Renner, Renato},
  journal = {Physical Review Letters},
  volume = {108},
  number = {20},
  pages = {200501},
  numpages = {5},
  year = {2012},
  month = {May},
  publisher = {American Physical Society},
  doi = {10.1103/PhysRevLett.108.200501},
  url = {https://link.aps.org/doi/10.1103/PhysRevLett.108.200501}
}

@article{Ben03,
title = {Notes on {L}andauer's principle, reversible computation, and {M}axwell's Demon},
journal = {Studies in History and Philosophy of Science Part B: Studies in History and Philosophy of Modern Physics},
volume = {34},
number = {3},
pages = {501-510},
year = {2003},
note = {{Q}uantum Information and Computation},
issn = {1355-2198},
doi = {https://doi.org/10.1016/S1355-2198(03)00039-X},
url = {https://www.sciencedirect.com/science/article/pii/S135521980300039X},
author = {Charles H. Bennett}
}

@article{GLTZ06,
  title = {Canonical Typicality},
  author = {Goldstein, Sheldon and Lebowitz, Joel L. and Tumulka, Roderich and Zangh\`{\i}, Nino},
  journal = {Physical Review Letters},
  volume = {96},
  issue = {5},
  pages = {050403},
  numpages = {3},
  year = {2006},
  month = {Feb},
  publisher = {American Physical Society},
  doi = {10.1103/PhysRevLett.96.050403},
  url = {https://link.aps.org/doi/10.1103/PhysRevLett.96.050403}
}

@article{ADHW09,
   title={The mother of all protocols: restructuring quantum information’s family tree},
   volume={465},
   ISSN={1471-2946},
   url={http://dx.doi.org/10.1098/rspa.2009.0202},
   DOI={10.1098/rspa.2009.0202},
   number={2108},
   journal={Proceedings of the Royal Society A: Mathematical, Physical and Engineering Sciences},
   publisher={The Royal Society},
   author={Abeyesinghe, Anura and Devetak, Igor and Hayden, Patrick and Winter, Andreas},
   year={2009},
   month=jun, pages={2537–2563} }

@article{MDS+13,
    author = {{M{\"u}ller-Lennert}, Martin and {Dupuis}, Fr{\'e}d{\'e}ric and {Szehr}, Oleg and {Fehr}, Serge and {Tomamichel}, Marco},
    title = {On quantum {R}{\'e}nyi entropies: A new generalization and some properties},
    journal = {Journal of Mathematical Physics},
    volume = {54},
    number = {12},
    pages = {122203},
    year = {2013},
    month = {12},
    issn = {0022-2488},
    doi = {10.1063/1.4838856},
    url = {https://doi.org/10.1063/1.4838856},
}

@article{WWY14,
  author  = {Mark M. Wilde and Andreas Winter and Dong Yang},
  title   = {Strong converse for the classical capacity of entanglement-breaking and {Hadamard} channels via a sandwiched {R\'enyi} relative entropy},
  journal = {Communications in Mathematical Physics},
  volume  = {331},
  number  = {2},
  pages   = {593--622},
  year    = {2014},
  month   = oct,
  doi     = {10.1007/s00220-014-2057-z},
  note    = {arXiv:1306.1586},
}

@article{CMW16,
  author  = {Tom Cooney and Mil{\'a}n Mosonyi and Mark M. Wilde},
  title   = {Strong converse exponents for a quantum channel discrimination problem and quantum‐feedback‐assisted communication},
  journal = {Communications in Mathematical Physics},
  volume  = {344},
  number  = {3},
  pages   = {797--829},
  year    = {2016},
  month   = jun,
  doi     = {10.1007/s00220-016-2645-4},
  note    = {arXiv:1408.3373},
}

@misc{BSST02,
      title={Entanglement-assisted capacity of a quantum channel and the reverse Shannon theorem}, 
      author={Charles H. Bennett and Peter W. Shor and John A. Smolin and Ashish V. Thapliyal},
      year={2002},
      eprint={quant-ph/0106052},
      archivePrefix={arXiv},
      primaryClass={quant-ph},
      url={https://arxiv.org/abs/quant-ph/0106052}, 
}

@article{FWT20,
   title={Quantum Channel Simulation and the Channel’s Smooth Max-Information},
   volume={66},
   ISSN={1557-9654},
   url={http://dx.doi.org/10.1109/TIT.2019.2943858},
   DOI={10.1109/tit.2019.2943858},
   number={4},
   journal={IEEE Transactions on Information Theory},
   publisher={Institute of Electrical and Electronics Engineers (IEEE)},
   author={Fang, Kun and Wang, Xin and Tomamichel, Marco and Berta, Mario},
   year={2020},
   month=apr, pages={2129–2140} }

@article{FGR25,
   title={Asymptotic Equipartition Theorems in von {N}eumann Algebras},
   ISSN={1424-0661},
   url={http://dx.doi.org/10.1007/s00023-025-01545-3},
   DOI={10.1007/s00023-025-01545-3},
   journal={Annales Henri Poincaré},
   publisher={Springer Science and Business Media LLC},
   author={Fawzi, Omar and Gao, Li and Rahaman, Mizanur},
   year={2025},
   month={February} }

@article{BRL+19,
   title={Thermodynamics as a Consequence of Information Conservation},
   volume={3},
   ISSN={2521-327X},
   url={http://dx.doi.org/10.22331/q-2019-02-14-121},
   DOI={10.22331/q-2019-02-14-121},
   journal={Quantum},
   publisher={Verein zur Forderung des Open Access Publizierens in den Quantenwissenschaften},
   author={Bera, Manabendra Nath and Riera, Arnau and Lewenstein, Maciej and Khanian, Zahra Baghali and Winter, Andreas},
   year={2019},
   month=feb, pages={121} }

@misc{Tas00,
      title={Jarzynski Relations for Quantum Systems and Some Applications}, 
      author={Hal Tasaki},
      year={2000},
      eprint={cond-mat/0009244},
      archivePrefix={arXiv},
      primaryClass={cond-mat.stat-mech},
      url={https://arxiv.org/abs/cond-mat/0009244}, 
}

@article{Ger93,
title = {Internal symmetries and limiting {G}ibbs states in quantum lattice mean-field theories},
journal = {Physica A: Statistical Mechanics and its Applications},
volume = {197},
number = {1},
pages = {284-300},
year = {1993},
issn = {0378-4371},
doi = {https://doi.org/10.1016/0378-4371(93)90474-I},
url = {https://www.sciencedirect.com/science/article/pii/037843719390474I},
author = {Thomas Gerisch},
}

@article{HP93,
author = {Hiai, Fumio and Petz, D\'{e}nes},
title = {ENTROPY DENSITIES FOR {G}IBBS STATES OF QUANTUM SPIN SYSTEMS},
journal = {Reviews in Mathematical Physics},
volume = {05},
number = {04},
pages = {693-712},
year = {1993},
doi = {10.1142/S0129055X93000218},

URL = { 
    
        https://doi.org/10.1142/S0129055X93000218

},
eprint = {        https://doi.org/10.1142/S0129055X93000218}
}

@article{MEC+19,
  title = {Optimal Probabilistic Work Extraction beyond the Free Energy Difference with a Single-Electron Device},
  author = {Maillet, Olivier and Erdman, Paolo A. and Cavina, Vasco and Bhandari, Bibek and Mannila, Elsa T. and Peltonen, Joonas T. and Mari, Andrea and Taddei, Fabio and Jarzynski, Christopher and Giovannetti, Vittorio and Pekola, Jukka P.},
  journal = {Physical Review Letters},
  volume = {122},
  issue = {15},
  pages = {150604},
  numpages = {6},
  year = {2019},
  month = {Apr},
  publisher = {American Physical Society},
  doi = {10.1103/PhysRevLett.122.150604},
  url = {https://link.aps.org/doi/10.1103/PhysRevLett.122.150604}
}

@article{DW19,
   title={Quantum Reading Capacity: General Definition and Bounds},
   volume={65},
   ISSN={1557-9654},
   url={http://dx.doi.org/10.1109/TIT.2019.2929925},
   DOI={10.1109/tit.2019.2929925},
   number={11},
   journal={IEEE Transactions on Information Theory},
   publisher={Institute of Electrical and Electronics Engineers (IEEE)},
   author={Das, Siddhartha and Wilde, Mark M.},
   year={2019},
   month=nov, pages={7566–7583} }

@article{MDP22,
   title={Quantum speed limits for information and coherence},
   volume={24},
   ISSN={1367-2630},
   url={http://dx.doi.org/10.1088/1367-2630/ac753c},
   DOI={10.1088/1367-2630/ac753c},
   number={6},
   journal={New Journal of Physics},
   publisher={IOP Publishing},
   author={Mohan, Brij and Das, Siddhartha and Pati, Arun Kumar},
   year={2022},
   month=jun, pages={065003} }

@misc{BC25,
      title={Quantum Spin Chains Thermalize at All Temperatures}, 
      author={Thiago Bergamaschi and Chi-Fang Chen},
      year={2025},
      eprint={2510.08533},
      archivePrefix={arXiv},
      primaryClass={quant-ph},
      url={https://arxiv.org/abs/2510.08533}, 
}

@misc{HSDW24,
      title={Cost of quantum secret key}, 
      author={Karol Horodecki and Leonard Sikorski and Siddhartha Das and Mark M. Wilde},
      year={2024},
      eprint={2402.17007},
      archivePrefix={arXiv},
      primaryClass={quant-ph},
      url={https://arxiv.org/abs/2402.17007}, 
}

@article{FBB21,
   title={Thermodynamic Implementations of Quantum Processes},
   volume={384},
   ISSN={1432-0916},
   url={http://dx.doi.org/10.1007/s00220-021-04107-w},
   DOI={10.1007/s00220-021-04107-w},
   number={3},
   journal={Communications in Mathematical Physics},
   publisher={Springer Science and Business Media LLC},
   author={Faist, Philippe and Berta, Mario and Brandao, Fernando G. S. L.},
   year={2021},
   month=may, pages={1709–1750} }

@article{CC22,
  title = {Thermodynamic Principle for Quantum Metrology},
  author = {Chu, Yaoming and Cai, Jianming},
  journal = {Physical Review Letters},
  volume = {128},
  issue = {20},
  pages = {200501},
  numpages = {6},
  year = {2022},
  month = {May},
  publisher = {American Physical Society},
  doi = {10.1103/PhysRevLett.128.200501},
  url = {https://link.aps.org/doi/10.1103/PhysRevLett.128.200501}
}

@article{PRY+22,
	author = {Proctor, Timothy and Rudinger, Kenneth and Young, Kevin and Nielsen, Erik and Blume-Kohout, Robin},
	date = {2022/01/01},
	doi = {10.1038/s41567-021-01409-7},
	id = {Proctor2022},
	isbn = {1745-2481},
	journal = {Nature Physics},
	number = {1},
	pages = {75--79},
	title = {Measuring the capabilities of quantum computers},
	url = {https://doi.org/10.1038/s41567-021-01409-7},
	volume = {18},
	year = {2022}}

@article{LJL+10,
   title={Quantum computers},
   volume={464},
   ISSN={1476-4687},
   url={http://dx.doi.org/10.1038/nature08812},
   DOI={10.1038/nature08812},
   number={7285},
   journal={Nature},
   publisher={Springer Science and Business Media LLC},
   author={Ladd, T. D. and Jelezko, F. and Laflamme, R. and Nakamura, Y. and Monroe, C. and O’Brien, J. L.},
   year={2010},
   month=mar, pages={45–53} }

@article{MCZG24,
  title = {Thermodynamics of the Quantum {M}pemba Effect},
  author = {Moroder, Mattia and Culhane, Ois\'{\i}n and Zawadzki, Krissia and Goold, John},
  journal = {Physical Review Letters},
  volume = {133},
  issue = {14},
  pages = {140404},
  numpages = {6},
  year = {2024},
  month = {Oct},
  publisher = {American Physical Society},
  doi = {10.1103/PhysRevLett.133.140404},
  url = {https://link.aps.org/doi/10.1103/PhysRevLett.133.140404}
}

@article{LKJ+15,
  title = {Quantum Coherence, Time-Translation Symmetry, and Thermodynamics},
  author = {Lostaglio, Matteo and Korzekwa, Kamil and Jennings, David and Rudolph, Terry},
  journal = {Physical Review X},
  volume = {5},
  issue = {2},
  pages = {021001},
  numpages = {11},
  year = {2015},
  month = {Apr},
  publisher = {American Physical Society},
  doi = {10.1103/PhysRevX.5.021001},
  url = {https://link.aps.org/doi/10.1103/PhysRevX.5.021001}
}

@article{SS21,
  title = {Quantum Thermodynamics of Correlated-Catalytic State Conversion at Small Scale},
  author = {Shiraishi, Naoto and Sagawa, Takahiro},
  journal = {Physical Review Letters},
  volume = {126},
  issue = {15},
  pages = {150502},
  numpages = {6},
  year = {2021},
  month = {Apr},
  publisher = {American Physical Society},
  doi = {10.1103/PhysRevLett.126.150502},
  url = {https://link.aps.org/doi/10.1103/PhysRevLett.126.150502}
}

@article{SSP15,
	author = {Skrzypczyk, Paul and Short, Anthony J. and Popescu, Sandu},
	date = {2014/06/27},
	doi = {10.1038/ncomms5185},
	id = {Skrzypczyk2014},
	isbn = {2041-1723},
	journal = {Nature Communications},
	number = {1},
	pages = {4185},
	title = {Work extraction and thermodynamics for individual quantum systems},
	url = {https://doi.org/10.1038/ncomms5185},
	volume = {5},
	year = {2014}}

@article{EV11,
   title={Second law and {L}andauer principle far from equilibrium},
   volume={95},
   ISSN={1286-4854},
   url={http://dx.doi.org/10.1209/0295-5075/95/40004},
   DOI={10.1209/0295-5075/95/40004},
   number={4},
   journal={EPL (Europhysics Letters)},
   publisher={IOP Publishing},
   author={Esposito, M. and Van den Broeck, C.},
   year={2011},
   month=aug, pages={40004} }

@article{FFRS20,
  title = {Chain Rule for the Quantum Relative Entropy},
  author = {Fang, Kun and Fawzi, Omar and Renner, Renato and Sutter, David},
  journal = {Physical Review Letters},
  volume = {124},
  issue = {10},
  pages = {100501},
  numpages = {6},
  year = {2020},
  month = {Mar},
  publisher = {American Physical Society},
  doi = {10.1103/PhysRevLett.124.100501},
  url = {https://link.aps.org/doi/10.1103/PhysRevLett.124.100501}
}

@article{LXD+25,
  title = {Temporal asymmetry in entanglement distillation},
  author = {Li, Yuhang and Xing, Junjing and Qu, Dengke and Gao, Huixia and Xiao, Lei and Liu, Jin-Ming and Xiao, Yunlong and Xue, Peng},
  journal = {Physical Review Letters},
  year = {2025},
  month = {Sep},
  publisher = {American Physical Society},
  doi = {10.1103/glc7-xy8t},
  url = {https://link.aps.org/doi/10.1103/glc7-xy8t}
}

@article{CTH09,
  title = {Fluctuation Theorem for Arbitrary Open Quantum Systems},
  author = {Campisi, Michele and Talkner, Peter and H\"anggi, Peter},
  journal = {Physical Review Letters},
  volume = {102},
  issue = {21},
  pages = {210401},
  numpages = {4},
  year = {2009},
  month = {May},
  publisher = {American Physical Society},
  doi = {10.1103/PhysRevLett.102.210401},
  url = {https://link.aps.org/doi/10.1103/PhysRevLett.102.210401}
}

@article{ULK15,
  title = {Equivalence of Quantum Heat Machines, and Quantum-Thermodynamic Signatures},
  author = {Uzdin, Raam and Levy, Amikam and Kosloff, Ronnie},
  journal = {Physical Review X},
  volume = {5},
  issue = {3},
  pages = {031044},
  numpages = {21},
  year = {2015},
  month = {Sep},
  publisher = {American Physical Society},
  doi = {10.1103/PhysRevX.5.031044},
  url = {https://link.aps.org/doi/10.1103/PhysRevX.5.031044}
}

@article{SU08,
  title = {Second Law of Thermodynamics with Discrete Quantum Feedback Control},
  author = {Sagawa, Takahiro and Ueda, Masahito},
  journal = {Physical Review Letters},
  volume = {100},
  issue = {8},
  pages = {080403},
  numpages = {4},
  year = {2008},
  month = {Feb},
  publisher = {American Physical Society},
  doi = {10.1103/PhysRevLett.100.080403},
  url = {https://link.aps.org/doi/10.1103/PhysRevLett.100.080403}
}

@misc{DGP24,
      title={Conditional entropy and information of quantum processes}, 
      author={Siddhartha Das and Kaumudibikash Goswami and Vivek Pandey},
      year={2024},
      note={arXiv:2410.01740},
      archivePrefix={arXiv},
      primaryClass={quant-ph},
      url={https://arxiv.org/abs/2410.01740}, 
}

@article{Yua19,
  title = {Hypothesis testing and entropies of quantum channels},
  author = {Yuan, Xiao},
  journal = {Physical Review A},
  volume = {99},
  number = {3},
  pages = {032317},
  numpages = {8},
  year = {2019},
  month = {March},
  publisher = {American Physical Society},
  doi = {10.1103/PhysRevA.99.032317},
  url = {https://link.aps.org/doi/10.1103/PhysRevA.99.032317}
}

@misc{Gou24,
      title={Resources of the Quantum World}, 
      author={Gilad Gour},
      year={2024},
      eprint={2402.05474},
      archivePrefix={arXiv},
      primaryClass={quant-ph},
      url={https://arxiv.org/abs/2402.05474}, 
}

@misc{SSC23,
      title={A resource theory of activity for quantum thermodynamics in the absence of heat baths}, 
      author={Swati and Uttam Singh and Giulio Chiribella},
      year={2023},
      eprint={2304.08926},
      archivePrefix={arXiv},
      primaryClass={quant-ph},
      url={https://arxiv.org/abs/2304.08926}, 
}

@article{Pek15,
	author = {Pekola, Jukka P. },
	date = {2015/02/01},
	doi = {10.1038/nphys3169},
	id = {Pekola2015},
	isbn = {1745-2481},
	journal = {Nature Physics},
	number = {2},
	pages = {118--123},
	title = {Towards quantum thermodynamics in electronic circuits},
	url = {https://doi.org/10.1038/nphys3169},
	volume = {11},
	year = {2015}}

@article{Ho04,
  title = {Universal Thermodynamics of Degenerate Quantum Gases in the Unitarity Limit},
  author = {Ho, Tin-Lun},
  journal = {Physical Review Letters},
  volume = {92},
  issue = {9},
  pages = {090402},
  numpages = {4},
  year = {2004},
  month = {Mar},
  publisher = {American Physical Society},
  doi = {10.1103/PhysRevLett.92.090402},
  url = {https://link.aps.org/doi/10.1103/PhysRevLett.92.090402}
}

@article{SDC21,
  title = {Partial order on passive states and {H}offman majorization in quantum thermodynamics},
  author = {Singh, Uttam and Das, Siddhartha and Cerf, Nicolas J.},
  journal = {Physical Review Research},
  volume = {3},
  issue = {3},
  pages = {033091},
  numpages = {23},
  year = {2021},
  month = {Jul},
  publisher = {American Physical Society},
  doi = {10.1103/PhysRevResearch.3.033091},
  url = {https://link.aps.org/doi/10.1103/PhysRevResearch.3.033091}
}

@article{SPSD25,
	author = {Sohail and Pandey, Vivek and Singh, Uttam and Das, Siddhartha},
	date = {2025/07/24},
	doi = {10.1007/s00023-025-01590-y},
	id = {Sohail2025},
	isbn = {1424-0661},
	journal = {Annales Henri Poincar{\'e}},
	title = {Fundamental Limitations on the Recoverability of Quantum Processes},
	url = {https://doi.org/10.1007/s00023-025-01590-y},
	year = {2025},
 note={arXiv:2403.12947}}

@article{HHHH09,
  title = {Quantum entanglement},
  author = {Horodecki, Ryszard and Horodecki, Pawe\l{} and Horodecki, Micha\l{} and Horodecki, Karol},
  journal = {Reviews of Modern Physics},
  volume = {81},
  issue = {2},
  pages = {865--942},
  numpages = {0},
  year = {2009},
  month = {Jun},
  publisher = {American Physical Society},
  doi = {10.1103/RevModPhys.81.865},
  url = {https://link.aps.org/doi/10.1103/RevModPhys.81.865}
}

@article{BCB+25,
  title = {Exploring {non-Markovianity} in ergodic channels: Measuring memory retention through ergotropy},
  author = {Basu, Ritam and Chakraborty, Anish and Badhani, Himanshu and Alimuddin, Mir and Bhattacharya, Samyadeb},
  journal = {Physical Review A},
  volume = {111},
  issue = {3},
  pages = {032416},
  numpages = {15},
  year = {2025},
  month = {Mar},
  publisher = {American Physical Society},
  doi = {10.1103/PhysRevA.111.032416},
  url = {https://link.aps.org/doi/10.1103/PhysRevA.111.032416}
}

@ARTICLE{Dat09,
  author={Datta, Nilanjana},
  journal={IEEE Transactions on Information Theory}, 
  title={Min- and Max-Relative Entropies and a New Entanglement Monotone}, 
  year={2009},
  volume={55},
  number={6},
  pages={2816-2826},
  doi={10.1109/TIT.2009.2018325}
}

@article{YHW19,
  title = {Distributed Private Randomness Distillation},
  author = {Yang, Dong and Horodecki, Karol and Winter, Andreas},
  journal = {Physical Review Letters},
  volume = {123},
  number = {17},
  pages = {170501},
  numpages = {6},
  year = {2019},
  month = {October},
  publisher = {American Physical Society},
  doi = {10.1103/PhysRevLett.123.170501},
  url = {https://link.aps.org/doi/10.1103/PhysRevLett.123.170501}
}

@misc{Tom21,
  author = {Marco Tomamichel},
  title = {Quantum Information Processing with Finite Resources -- Mathematical Foundations},
  year = {2021},
  note = {arXiv:1504.00233v5},
  archivePrefix = {arXiv},
  primaryClass = {quant-ph},
  doi = {10.48550/arXiv.1504.00233},
  url = {https://arxiv.org/abs/1504.00233}
}

@misc{YZGZ20,
      title={One-shot dynamical resource theory}, 
      author={Xiao Yuan and Pei Zeng and Minbo Gao and Qi Zhao},
      year={2020},
      eprint={2012.02781},
      archivePrefix={arXiv},
      primaryClass={quant-ph},
      url={https://arxiv.org/abs/2012.02781}, 
}

@article{LY20,
   title={Operational resource theory of quantum channels},
   volume={2},
   ISSN={2643-1564},
   url={http://dx.doi.org/10.1103/PhysRevResearch.2.012035},
   DOI={10.1103/physrevresearch.2.012035},
   number={1},
   journal={Physical Review Research},
   publisher={American Physical Society (APS)},
   author={Liu, Yunchao and Yuan, Xiao},
   year={2020},
   month=feb }

@misc{LW19,
      title={Resource theories of quantum channels and the universal role of resource erasure}, 
      author={Zi-Wen Liu and Andreas Winter},
      year={2019},
      eprint={1904.04201},
      archivePrefix={arXiv},
      primaryClass={quant-ph},
      url={https://arxiv.org/abs/1904.04201}, 
}

@article{KDWW19,
  title = {Extendibility Limits the Performance of Quantum Processors},
  author = {Kaur, Eneet and Das, Siddhartha and Wilde, Mark M. and Winter, Andreas},
  journal = {Physical Review Letters},
  volume = {123},
  number = {7},
  pages = {070502},
  numpages = {7},
  year = {2019},
  month = {August},
  publisher = {American Physical Society},
  doi = {10.1103/PhysRevLett.123.070502},
  url = {https://link.aps.org/doi/10.1103/PhysRevLett.123.070502}
}

@article{DBWH21,
  title = {Universal Limitations on Quantum Key Distribution over a Network},
  author = {Das, Siddhartha and B\"auml, Stefan and Winczewski, Marek and Horodecki, Karol},
  journal = {Physical Review X},
  volume = {11},
  number = {4},
  pages = {041016},
  numpages = {38},
  year = {2021},
  month = {October},
  publisher = {American Physical Society},
  doi = {10.1103/PhysRevX.11.041016},
  url = {https://link.aps.org/doi/10.1103/PhysRevX.11.041016}
}

@article{RT21,
  title = {One-Shot Manipulation of Dynamical Quantum Resources},
  author = {Regula, Bartosz and Takagi, Ryuji},
  journal = {Physical Review Letters},
  volume = {127},
  issue = {6},
  pages = {060402},
  numpages = {9},
  year = {2021},
  month = {Aug},
  publisher = {American Physical Society},
  doi = {10.1103/PhysRevLett.127.060402},
  url = {https://link.aps.org/doi/10.1103/PhysRevLett.127.060402}
}

@article{DKSW18,
   title={Fundamental limits on quantum dynamics based on entropy change},
   volume={59},
   ISSN={1089-7658},
   url={http://dx.doi.org/10.1063/1.4997044},
   DOI={10.1063/1.4997044},
   number={1},
   journal={Journal of Mathematical Physics},
   publisher={AIP Publishing},
   author={Das, Siddhartha and Khatri, Sumeet and Siopsis, George and Wilde, Mark M.},
   year={2018},
   month=jan }

@misc{Das19,
      title={Bipartite Quantum Interactions: Entangling and Information Processing Abilities}, 
      author={Siddhartha Das},
      year={2019},
      note={arXiv:1901.05895},
      archivePrefix={arXiv},
      primaryClass={quant-ph},
      url={https://arxiv.org/abs/1901.05895}, 
}

@misc{SSHD24,
      title={Practical limitations on robustness and scalability of quantum Internet}, 
      author={Abhishek Sadhu and Meghana Ayyala Somayajula and Karol Horodecki and Siddhartha Das},
      year={2024},
      eprint={2308.12739},
      archivePrefix={arXiv},
      primaryClass={quant-ph},
      url={https://arxiv.org/abs/2308.12739}, 
}

@article{G19,
	doi = {10.1109/tit.2019.2907989},
  	url = {https://doi.org/10.1109%2Ftit.2019.2907989},
  	year = 2019,
	month = {September},
  	publisher = {Institute of Electrical and Electronics Engineers ({IEEE})},
  	volume = {65},
  	number = {9},
  	pages = {5880--5904},
  	author = {Gilad Gour},
  	title = {Comparison of Quantum Channels by Superchannels},
  	journal = {{IEEE} Transactions on Information Theory}
}

@article{DJKR06,
   title={Multiplicativity of Completely Bounded p-Norms Implies a New Additivity Result},
   volume={266},
   ISSN={1432-0916},
   url={http://dx.doi.org/10.1007/s00220-006-0034-0},
   DOI={10.1007/s00220-006-0034-0},
   number={1},
   journal={Communications in Mathematical Physics},
   publisher={Springer Science and Business Media LLC},
   author={Devetak, Igor and Junge, Marius and King, Christoper and Ruskai, Mary Beth},
   year={2006},
   month=may, pages={37–63} }

@article{LKDW18,
  title = {Approaches for approximate additivity of the {H}olevo information of quantum channels},
  author = {Leditzky, Felix and Kaur, Eneet and Datta, Nilanjana and Wilde, Mark M.},
  journal = {Physical Review A},
  volume = {97},
  issue = {1},
  pages = {012332},
  numpages = {20},
  year = {2018},
  month = {January},
  publisher = {American Physical Society},
  url = {https://journals.aps.org/pra/abstract/10.1103/PhysRevA.97.012332}
}

@misc{LMB25,
      title={Thermodynamic criteria for signaling in quantum channels}, 
      author={Yutong Luo and Simon Milz and Felix C. Binder},
      year={2025},
      eprint={2506.20428},
      archivePrefix={arXiv},
      primaryClass={quant-ph},
      url={https://arxiv.org/abs/2506.20428}, 
}

@article{Len78,
	author = {Lenard, Andrew},
	date = {1978/12/01},
	doi = {10.1007/BF01011769},
	id = {Lenard1978},
	isbn = {1572-9613},
	journal = {Journal of Statistical Physics},
	number = {6},
	pages = {575--586},
	title = {Thermodynamical proof of the {G}ibbs formula for elementary quantum systems},
	url = {https://doi.org/10.1007/BF01011769},
	volume = {19},
	year = {1978}}

@book{PB21,
  title={Statistical Mechanics: International Series of Monographs in Natural Philosophy},
  author={Pathria, Raj K. and Beale, Paul D.},
  isbn={978-0-08-102692-2},
  year={2021},
  publisher={Elsevier},
doi = {10.1016/c2017-0-01713-5},
url = {http://dx.doi.org/10.1016/C2017-0-01713-5},
}

@article{Jay57a,
  title = {Information Theory and Statistical Mechanics},
  author = {Jaynes, E. T.},
  journal = {Physical Review},
  volume = {106},
  issue = {4},
  pages = {620--630},
  numpages = {0},
  year = {1957},
  month = {May},
  publisher = {American Physical Society},
  doi = {10.1103/PhysRev.106.620},
  url = {https://link.aps.org/doi/10.1103/PhysRev.106.620}
}

@article{TBB25,
  title = {Quantum Thermal Analogs of Electric Circuits: A Universal Approach},
  author = {Tiwari, Devvrat and Bhattacharya, Samyadeb and Banerjee, Subhashish},
  journal = {Physical Review Letters},
  volume = {135},
  issue = {2},
  pages = {020404},
  numpages = {6},
  year = {2025},
  month = {Jul},
  publisher = {American Physical Society},
  doi = {10.1103/5x8m-bhgd},
  url = {https://link.aps.org/doi/10.1103/5x8m-bhgd}
}

@misc{Pot24,
      title={Quantum Thermodynamics}, 
      author={Patrick P. Potts},
      year={2024},
      eprint={2406.19206},
      archivePrefix={arXiv},
      primaryClass={quant-ph},
      url={https://arxiv.org/abs/2406.19206}, 
}

@article{BGG+25,
   title={Causal and Noncausal Revivals of Information: A New Regime of Non-{Markovianity} in Quantum Stochastic Processes},
   volume={6},
   ISSN={2691-3399},
   url={http://dx.doi.org/10.1103/PRXQuantum.6.020316},
   DOI={10.1103/prxquantum.6.020316},
   number={2},
   journal={PRX Quantum},
   publisher={American Physical Society (APS)},
   author={Buscemi, Francesco and Gangwar, Rajeev and Goswami, Kaumudibikash and Badhani, Himanshu and Pandit, Tanmoy and Mohan, Brij and Das, Siddhartha and Bera, Manabendra Nath},
   year={2025},
   month=apr }

@article{DW05,
   title={Distillation of secret key and entanglement from quantum states},
   volume={461},
   ISSN={1471-2946},
   url={http://dx.doi.org/10.1098/rspa.2004.1372},
   DOI={10.1098/rspa.2004.1372},
   number={2053},
   journal={Proceedings of the Royal Society A: Mathematical, Physical and Engineering Sciences},
   publisher={The Royal Society},
   author={Devetak, Igor and Winter, Andreas},
   year={2005},
   month=jan, pages={207–235} }

@article{GMN+15,
   title={The resource theory of informational nonequilibrium in thermodynamics},
   volume={583},
   ISSN={0370-1573},
   url={http://dx.doi.org/10.1016/j.physrep.2015.04.003},
   DOI={10.1016/j.physrep.2015.04.003},
   journal={Physics Reports},
   publisher={Elsevier BV},
   author={Gour, Gilad and Müller, Markus P. and Narasimhachar, Varun and Spekkens, Robert W. and Yunger Halpern, Nicole},
   year={2015},
   month=jul, pages={1–58} }

@article{LPSW09,
  title = {Quantum mechanical evolution towards thermal equilibrium},
  author = {Linden, Noah and Popescu, Sandu and Short, Anthony J. and Winter, Andreas},
  journal = {Physical Review E},
  volume = {79},
  issue = {6},
  pages = {061103},
  numpages = {12},
  year = {2009},
  month = {Jun},
  publisher = {American Physical Society},
  doi = {10.1103/PhysRevE.79.061103},
  url = {https://link.aps.org/doi/10.1103/PhysRevE.79.061103}
}

@article{MHK23,
  title = {Universality classes of thermalization for mesoscopic {Floquet} systems},
  author = {Morningstar, Alan and Huse, David A. and Khemani, Vedika},
  journal = {Physical Review B},
  volume = {108},
  issue = {17},
  pages = {174303},
  numpages = {16},
  year = {2023},
  month = {Nov},
  publisher = {American Physical Society},
  doi = {10.1103/PhysRevB.108.174303},
  url = {https://link.aps.org/doi/10.1103/PhysRevB.108.174303}
}

@article{HMG19,
	author = {Hallam, A. and Morley, J. G. and Green, A. G.},
	date = {2019/06/20},
	doi = {10.1038/s41467-019-10336-4},
	id = {Hallam2019},
	isbn = {2041-1723},
	journal = {Nature Communications},
	number = {1},
	pages = {2708},
	title = {The {L}yapunov spectra of quantum thermalisation},
	url = {https://doi.org/10.1038/s41467-019-10336-4},
	volume = {10},
	year = {2019}}

@article{WBHK20,
   title={Amortized channel divergence for asymptotic quantum channel discrimination},
   volume={110},
   ISSN={1573-0530},
   url={http://dx.doi.org/10.1007/s11005-020-01297-7},
   DOI={10.1007/s11005-020-01297-7},
   number={8},
   journal={Letters in Mathematical Physics},
   publisher={Springer Science and Business Media LLC},
   author={Wilde, Mark M. and Berta, Mario and Hirche, Christoph and Kaur, Eneet},
   year={2020},
   month=jun, pages={2277–2336} }

@misc{BGD25b,
      title={Thermodynamic work capacity of quantum information processing}, 
      author={Himanshu Badhani and Dhanuja G S and Siddhartha Das},
      year={2025},
      eprint={2510.23731},
      archivePrefix={arXiv},
      primaryClass={quant-ph},
      url={https://arxiv.org/abs/2510.23731}, 
}

@article{BGC+25,
      author={Badhani, Himanshu and G S, Dhanuja and Choudhary, Swati and Anand, Vishal and Das, Siddhartha},
	title={Erasure cost of a quantum process: A thermodynamic meaning of the dynamical min-entropy},
	journal={Quantum Science and Technology},
	url={http://iopscience.iop.org/article/10.1088/2058-9565/ae34e2},
	year={2026},
      note={arXiv:2506.05307},
}

@misc{DS25,
      title={Maximum entropy principle for quantum processes}, 
      author={Siddhartha Das and Ujjwal Sen},
      year={2025},
      eprint={2506.24079},
      archivePrefix={arXiv},
      primaryClass={quant-ph},
      url={https://arxiv.org/abs/2506.24079}, 
}

@article{CDP09,
  title = {Theoretical framework for quantum networks},
  author = {Chiribella, Giulio and D'Ariano, Giacomo Mauro and Perinotti, Paolo},
  journal = {Physical Review A},
  volume = {80},
  issue = {2},
  pages = {022339},
  numpages = {20},
  year = {2009},
  month = {August},
  publisher = {American Physical Society},
  doi = {10.1103/PhysRevA.80.022339},
  url = {https://link.aps.org/doi/10.1103/PhysRevA.80.022339}
}

@article{BHO+13,
  title = {Resource Theory of Quantum States Out of Thermal Equilibrium},
  author = {Brand\~ao, Fernando G. S. L. and Horodecki, Micha\l{} and Oppenheim, Jonathan and Renes, Joseph M. and Spekkens, Robert W.},
  journal = {Physical Review Letters},
  volume = {111},
  issue = {25},
  pages = {250404},
  numpages = {5},
  year = {2013},
  month = {Dec},
  publisher = {American Physical Society},
  doi = {10.1103/PhysRevLett.111.250404},
  url = {https://link.aps.org/doi/10.1103/PhysRevLett.111.250404}
}

@article{BHN+15,
  title        = {The second laws of quantum thermodynamics},
  author       = {Brand{\~a}o, Fernando G. S. L. and Horodecki, Micha{\l} and Ng, Nelly H. Y. and Oppenheim, Jonathan and Wehner, Stephanie},
  journal      = {Proceedings of the National Academy of Sciences},
  volume       = {112},
  number       = {11},
  pages        = {3275--3279},
  year         = {2015},
  month        = {February},
  publisher    = {Proceedings of the National Academy of Sciences},
  doi          = {10.1073/pnas.1411728112},
  url          = {http://dx.doi.org/10.1073/pnas.1411728112},
  issn         = {1091-6490}
}

@article{BCGM24,
   title={Unified Framework for Continuity of Sandwiched {R{\'e}nyi} Divergences},
   ISSN={1424-0661},
   url={http://dx.doi.org/10.1007/s00023-024-01519-x},
   DOI={10.1007/s00023-024-01519-x},
   journal={Annales Henri Poincar{\'e}},
   publisher={Springer Science and Business Media LLC},
   author={Andreas Bluhm and {\'A}ngela Capel and Paul Gondolf and Tim M{\"o}bus},
   year={2024},
   month=dec }

@article{Sch96,
  title = {Sending entanglement through noisy quantum channels},
  author = {Schumacher, Benjamin},
  journal = {Physical Review A},
  volume = {54},
  issue = {4},
  pages = {2614--2628},
  numpages = {0},
  year = {1996},
  month = {Oct},
  publisher = {American Physical Society},
  doi = {10.1103/PhysRevA.54.2614},
  url = {https://link.aps.org/doi/10.1103/PhysRevA.54.2614}
}

@article{DW19b,
  title = {Quantum rebound capacity},
  author = {Das, Siddhartha and Wilde, Mark M.},
  journal = {Physical Review A},
  volume = {100},
  issue = {3},
  pages = {030302},
  numpages = {6},
  year = {2019},
  month = {Sep},
  publisher = {American Physical Society},
  doi = {10.1103/PhysRevA.100.030302},
  url = {https://link.aps.org/doi/10.1103/PhysRevA.100.030302}
}

@misc{BKSD23,
      title={Infinite Dimensional Asymmetric Quantum Channel Discrimination}, 
      author={Bjarne Bergh and Jan Kochanowski and Robert Salzmann and Nilanjana Datta},
      year={2023},
      eprint={2308.12959},
      archivePrefix={arXiv},
      primaryClass={quant-ph},
      url={https://arxiv.org/abs/2308.12959}, 
}

@ARTICLE{BLT25,
  author={Berta, Mario and Lami, Ludovico and Tomamichel, Marco},
  journal={IEEE Transactions on Information Theory}, 
  title={Continuity of Entropies via Integral Representations}, 
  year={2025},
  volume={71},
  number={3},
  pages={1896-1908},
  doi={10.1109/TIT.2025.3527858}}

@article{ABD+25,
  TITLE = {Continuity bounds for quantum entropies arising from a fundamental entropic inequality},
  AUTHOR = {Audenaert, Koenraad and Bergh, Bjarne and Datta, Nilanjana and Jabbour, Michael G and Capel, {\'A}ngela and Gondolf, Paul},
  URL = {https://hal.science/hal-05022361},
  JOURNAL = {{IEEE Transactions on Information Theory}},
  PUBLISHER = {{Institute of Electrical and Electronics Engineers}},
  VOLUME = {71},
  NUMBER = {9},
  PAGES = {7029 - 7038},
  YEAR = {2025},
  MONTH = Jul,
  DOI = {10.1109/TIT.2025.3586478} }

@article{Rus02,
  title={Inequalities for quantum entropy: A review with conditions for equality},
  author={Ruskai, Mary Beth},
  journal={Journal of Mathematical Physics},
  volume={43},
  number={9},
  pages={4358--4375},
  year={2002},
  publisher={American Institute of Physics}
}

@misc{DC19,
      title={Quantum Thermodynamics: An introduction to the thermodynamics of quantum information}, 
      author={Sebastian Deffner and Steve Campbell},
      year={2019},
      eprint={1907.01596},
      archivePrefix={arXiv},
      primaryClass={quant-ph},
      url={https://arxiv.org/abs/1907.01596}, 
}

@misc{ST25,
      title={Recovery of the second law in fully quantum thermodynamics}, 
      author={Naoto Shiraishi and Ryuji Takagi},
      year={2025},
      eprint={2510.05642},
      archivePrefix={arXiv},
      primaryClass={quant-ph},
      url={https://arxiv.org/abs/2510.05642}, 
}

@misc{HWS+25,
      title={Quantification of the energy consumption of entanglement distribution}, 
      author={Karol Horodecki and Marek Winczewski and Leonard Sikorski and Paweł Mazurek and Mikołaj Czechlewski and Raja Yehia},
      year={2025},
      eprint={2507.23108},
      archivePrefix={arXiv},
      primaryClass={quant-ph},
      url={https://arxiv.org/abs/2507.23108}, 
}

@article{GPW05,
  title = {Quantum, classical, and total amount of correlations in a quantum state},
  author = {Groisman, Berry and Popescu, Sandu and Winter, Andreas},
  journal = {Physical Review A},
  volume = {72},
  issue = {3},
  pages = {032317},
  numpages = {11},
  year = {2005},
  month = {Sep},
  publisher = {American Physical Society},
  doi = {10.1103/PhysRevA.72.032317},
  url = {https://link.aps.org/doi/10.1103/PhysRevA.72.032317}
}

@article{ELV10,
doi = {10.1088/1367-2630/12/1/013013},
url = {https://doi.org/10.1088/1367-2630/12/1/013013},
year = {2010},
month = {jan},
publisher = {},
volume = {12},
number = {1},
pages = {013013},
author = {Esposito, Massimiliano and Lindenberg, Katja and Van den Broeck, Christian},
title = {Entropy production as correlation between system and reservoir},
journal = {New Journal of Physics}
}
\end{document}